%% file: sample-sigconf.tex
\documentclass[sigconf]{acmart}

\usepackage{graphicx}
\usepackage{balance}
\usepackage{subfigure}

\usepackage{amsmath}
\usepackage{mathtools}
\usepackage{amsthm}
\usepackage{multirow}
\usepackage{colortbl}
\usepackage{algorithm,algorithmic}
\usepackage[capitalize,noabbrev]{cleveref}

\usepackage{mathtools}
\usepackage{tikz}
\input{preamble}

\AtBeginDocument{%
  }

\copyrightyear{2024}
\acmYear{2024}
\setcopyright{acmlicensed}\acmConference[MM '24]{Proceedings of the 32nd ACM International Conference on Multimedia}{October 28-November 1, 2024}{Melbourne, VIC, Australia}
\acmBooktitle{Proceedings of the 32nd ACM International Conference on Multimedia (MM '24), October 28-November 1, 2024, Melbourne, VIC, Australia}
\acmDOI{10.1145/3664647.3681556}
\acmISBN{979-8-4007-0686-8/24/10}
\begin{document}

\title{Diffusion Posterior Proximal Sampling for Image Restoration}

\author{Hongjie Wu}  \orcid{0009-0007-3203-1521}
 \email{wuhongjie0818@gmail.com}
\affiliation{
 \institution{College of Computer Science, Sichuan University} 
 \city{Chengdu} \country{China}
 }
 \additionalaffiliation{\institution{Engineering Research Center of Machine Learning and Industry Intelligence, Ministry of Education, China}}
 \authornote{Equal contribution.}

\author{Linchao He}   \orcid{0000-0002-0562-3026}
 \email{hlc@stu.scu.edu.cn}
\affiliation{
 \institution{National Key Laboratory of Fundamental Science on Synthetic Vision, Sichuan University} \city{Chengdu} \country{China}
 }
  \additionalaffiliation{\institution{College of Computer Science, Sichuan University}}
\authornotemark[2]

\author{Mingqin Zhang}  \orcid{0009-0008-7273-9302}  \authornotemark[1] 
 \email{zhangmingqin@stu.scu.edu.cn}
 \affiliation{%
\institution{College of Computer Science, Sichuan University}
 \city{Chengdu} \country{China}
 }

\author{Dongdong Chen}  \orcid{0000-0002-7016-9288}
\email{d.chen@hw.ac.uk}
\affiliation{%
  \institution{Heriot-Watt University} \city{Edinburgh}
  \country{United Kingdom}
}

\author{Kunming Luo}  \orcid{0000-0002-5070-7392}
\email{kluoad@connect.ust.hk}
\affiliation{%
  \institution{Hong Kong University of Science and Technology}
  \city{Hong Kong}
  \country{China}
}

 \author{Mengting Luo} \authornotemark[3]  \orcid{0000-0003-2848-1849}
\affiliation{
 \institution{National Key Laboratory of Fundamental Science on Synthetic Vision, Sichuan University} \city{Chengdu} \country{China}
 }

  \author{Ji-Zhe Zhou} \authornotemark[1]
   \email{jzzhou@scu.edu.cn} 
   \orcid{0000-0002-2447-1806}
\affiliation{
\institution{College of Computer Science, Sichuan University} 
 \city{Chengdu} \country{China}
 }

\author{Hu Chen}  \orcid{0000-0001-9300-6572}
\email{huchen@scu.edu.cn}
\affiliation{
\institution{College of Computer Science, Sichuan University} 
 \city{Chengdu} \country{China}
 } 

\author{Jiancheng Lv} \authornotemark[1]   \orcid{0000-0001-6551-3884}
 \email{lvjiancheng@scu.edu.cn}
\authornote{Corresponding author.}
\affiliation{
\institution{College of Computer Science, Sichuan University} 
 \city{Chengdu}
 \country{China}
 }

\renewcommand\shortauthors{Hongjie Wu et al.}


\begin{abstract}
Diffusion models have demonstrated remarkable efficacy in generating high-quality samples.
Existing diffusion-based image restoration algorithms exploit pre-trained diffusion models to leverage data priors, yet they still preserve elements inherited from the unconditional generation paradigm.
These strategies initiate the denoising process with pure white noise and incorporate random noise at each generative step, leading to over-smoothed results.
In this paper, we present a refined paradigm for diffusion-based image restoration. Specifically, we opt for a sample consistent with the measurement identity at each generative step, exploiting the sampling selection as an avenue for output stability and enhancement. The number of candidate samples used for selection is adaptively determined based on the signal-to-noise ratio of the timestep.
Additionally, we start the restoration process with an initialization combined with the measurement signal, providing supplementary information to better align the generative process.
Extensive experimental results and analyses validate that our proposed method significantly enhances image restoration performance while consuming negligible additional computational resources.
\end{abstract}




\maketitle

\section{Introduction}
\label{sec:Introduction}

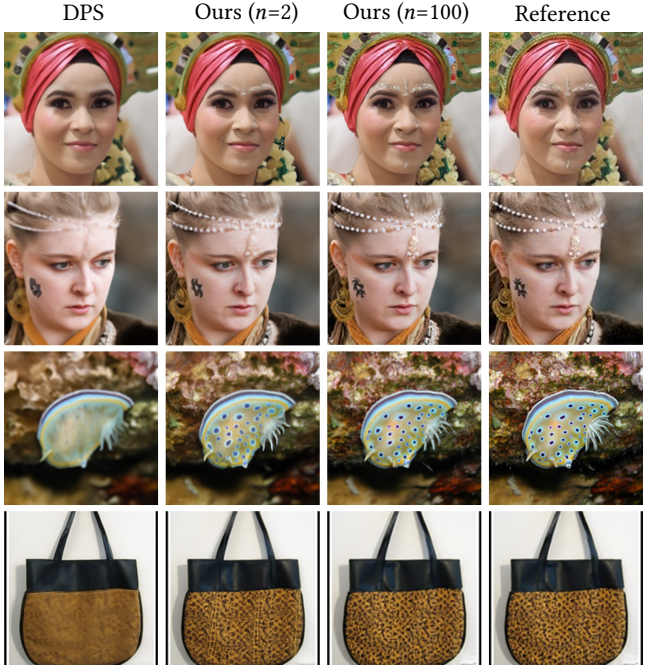
\begin{figure}[ht]
  \centering
  \begin{tikzpicture}
        \node[anchor=south west,inner sep=0] (image) at (0,0) {\includegraphics[width=\linewidth]{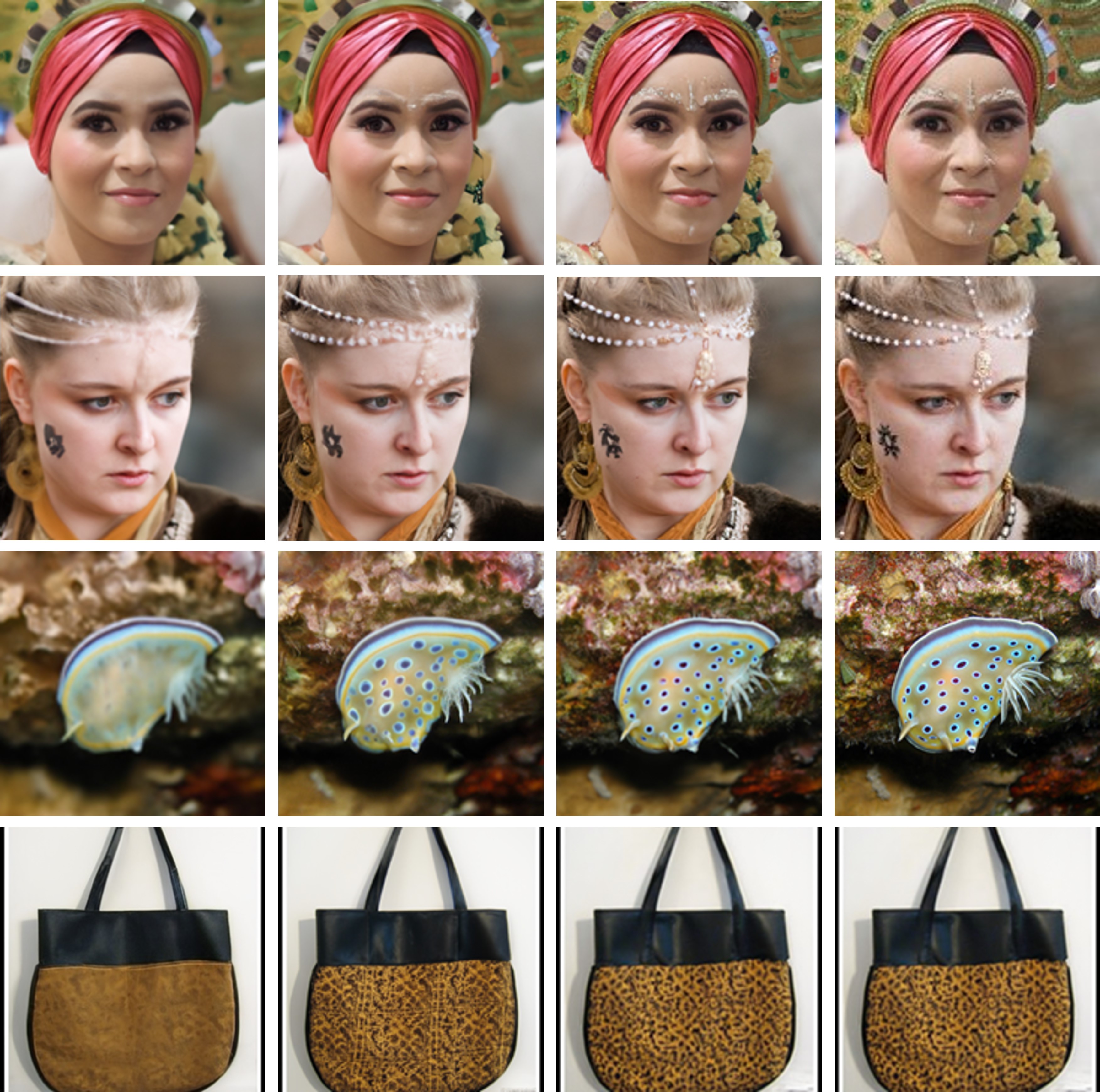}};
        \begin{scope}[x={(image.south east)},y={(image.north west)}]
            \node at (0.12,1.03) {{ DPS}};
            \node at (0.375,1.03) {{ Ours ($n$=2)}};
            \node at (0.625,1.03) {{ Ours ($n$=100)}};
            \node at (0.875,1.03) {{Reference}};
        \end{scope}
    \end{tikzpicture}
  \caption{Examples of the super-resolution ($\times$4) task to illustrate the efficiency of our method, where $n$ denotes the number of candidate samples, and DPS also refers to the case where $n=1$.}
  \Description{Examples of the super-resolution ($\times$4) task to illustrate the efficiency of our method, where $n$ denotes the number of candidate samples, and DPS also refers to the case where $n=1$.}
  \label{fig:n_number}
  \vspace{-0.3cm}
\end{figure}

\begin{figure*}[t]
\centering  \label{fig:main}
\includegraphics[width=1.\textwidth]{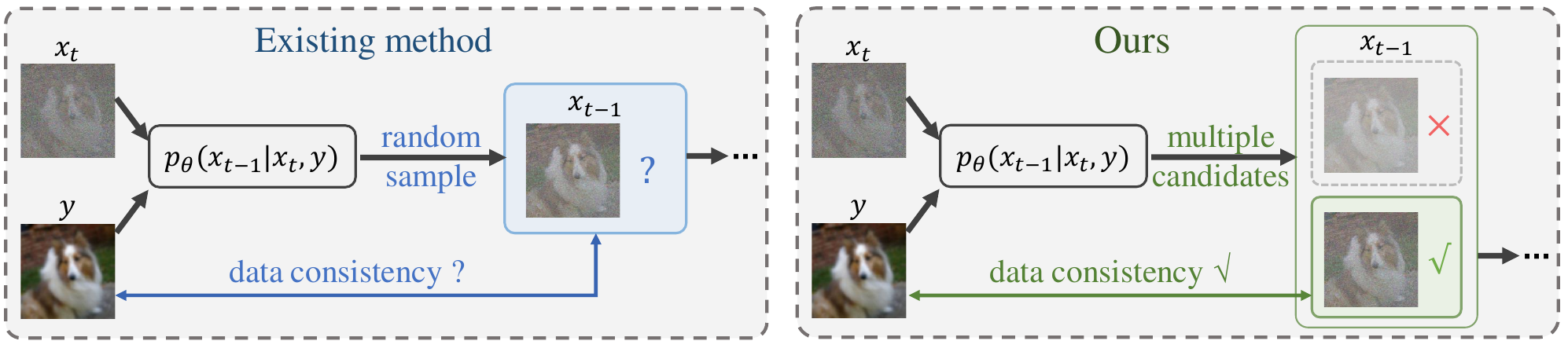}
\caption{High-level illustration of our proposed DPPS ($n$=2). Existing methods employ random sampling from the predicted distribution, our approach takes multiple samples and selecting the one with better data consistency (denoted by \textcolor{green}{$\checkmark$}) at each step.} \Description{High-level illustration of our proposed DPPS ($n$=2). Existing methods employ random sampling from the predicted distribution, our approach takes multiple samples and selecting the one with better data consistency (denoted by \textcolor{green}{$\checkmark$}) at each step.}
\vspace{-0.3cm}
\end{figure*}

Diffusion models~\cite{vincent2011connection,ho2020denoising,song2021scorebased} have gained broad attention due to the powerful ability to model complex data distribution, and have been applied to a wide range of tasks such as image generation~\cite{dhariwal2021diffusion,meng2021sdedit,rombach2022high}, molecule generation~\cite{corso2023diffdock}, and natural language processing~\cite{li2022diffusion}.
Recent research has demonstrated that pre-trained unconditional diffusion models can be effectively employed to address image restoration problems~\cite{song2022solving,kawar2022denoising,chung2022improving,wang2023zeroshot,song2023pseudoinverseguided} to leverage rich priors in a plug-and-play fashion, and achieve significant advancements.

However, for diffusion-based image restoration algorithms, they still retain a process inherited from unconditional generation. In particular, (i) these methods start generating the image that corresponds to the measurement with a white noise as initialization, and (ii) a fully stochastic noise is introduced at each step of the generation process.
We argue that this paradigm is inappropriate for solving image restoration tasks. Firstly, randomness enriches the diversity of the generated samples in unconditional generation~\cite{liu2020diverse,xiao2021tackling}. However, for image restoration tasks where the identity information of the measurement needs to be preserved~\cite{yao2024building}, randomness leads to uncontrollable generation outcomes~\cite{cao2023deep}.
Secondly, in the generative process of diffusion-based model, the current state is randomly sampled~\cite{chung2023diffusion,song2023pseudoinverseguided} from the predicted distribution without any correction.
Considering the Markovian reverse process, the sampling quality of the next state heavily depends on the current sampling result. If the current result deviates considerably from expectations, the subsequent generative process will encounter a significant exposure bias problem~\cite{ning2023input,li2023alleviating}.
As a consequence, randomness introduces fluctuations and instability into the generative process, ultimately leading to over-smoothed results.

In this paper, we advance toward a more refined paradigm for solving image inverse problems with pre-trained diffusion models. 
Specifically, to ensure the sampling result is closely consistent with the measurement identity, we opt for the proximal sample at each step from multiple candidate samples, as opposed to making a random choice. 
By mitigating uncertainties, the sample consistently progresses toward the desired targets, resulting in a stable generative process and improved output quality. Moreover,
we start the generation process with an initialization composed of both the measurement signal and white noise, rather than pure white noise. Consequently, the denoising process commences from a closer starting point, facilitating a faster convergence toward the desired result.

We theoretically analyze that our approach reduces variance in comparison to existing state-of-the-art methods which introduce random sampling. 
And we conduct extensive experiments to demonstrate the superior restoration capabilities of our method across diverse image restoration tasks, such as super-resolution (SR), deblurring, and inpainting, with only a marginal additional cost. Furthermore, sufficient experimental analyses are shown to validate the effectiveness of our proposed strategy.

In summary, our contributions are as follows:

\begin{itemize}
\item We pioneer to improve the generation quality by exploiting sampling choice from the predicted distribution of each reverse step.
\item We propose an efficient proximal sampling strategy that aligns with measurement identity to solve diffusion-based image restoration problems.
\item Extensive experimental results demonstrate that our proposed method outperforms other state-of-the-art methods in image restoration performance, requiring only minimal additional computation.

\end{itemize}

\section{Background}
\label{sec:Background}
\subsection{Denoising Diffusion Probabilistic Models}
Denoising Diffusion Probabilistic Models (DDPMs)~\cite{ho2020denoising} is a class of models that incorporate a forward (diffusion) process and a corresponding reverse (generative) process.
Consider a $T$-step forward process,
the noised sample $\mathbf{x}_t$ at time step $t$ can be modeled from the previous state $\mathbf{x}_{t-1}$:
\begin{equation}\label{eq:ddpmforward1}
    q(\mathbf{x}_{t}|\mathbf{x}_{t-1})=\mathcal{N}(\mathbf{x}_{t};\sqrt{1-\beta_{t}}\mathbf{x}_{t-1},\beta_{t}\mathbf{I}),    
\end{equation}
where $\mathcal{N}$ denotes the Gaussian distribution, and $\beta_{t}$ is a pre-defined parameter increasing with $t$.
Through reparameterization, $\mathbf{x}_t$ satisfies the following conditional distribution $q(\mathbf{x}_{t}|\mathbf{x}_{0})=\mathcal{N}(\mathbf{x}_{t};\sqrt{\Bar{\alpha}_{t}}\mathbf{x}_{0},\\(1-\Bar{\alpha}_{t})\mathbf{I})$, 
where $\Bar{\alpha}_{t} = \prod_{i=0}^{t}\alpha_{i}$ and $\alpha_{t} = 1- \beta_{t}$.

On the other hand, the reverse process of DDPM is formulated by the transition kernel parameterized by $\thetab$:
\begin{equation}\label{eq:ddpmreverse}
    p_\thetab(\mathbf{x}_{t-1}|\mathbf{x}_{t})=\mathcal{N}(\mathbf{x}_{t-1}; \mub_{\thetab}(\x_{t},t),\sigma_{t}^2 \mathbf{I}).
\end{equation}
where the learning objective aims to minimize the KL divergence between the forward and backward processes. The epsilon-matching objective is typically set as:
\begin{equation}\label{eq:epsilon-matching}
    \mathbf{L}({\thetab}) = \mathbb{E}_{\mathbf{x}_0 \sim q(\mathbf{x}_0), {\epsilonb} \sim {\Nc} (\mathbf{0}, \mathbf{I}), t \sim \mathbb{U}(\{1,...,T\})} [\|{\epsilonb} - {\epsilonb}_{\mathbf{\theta}}(\mathbf{x}_t,t)  \|^2].
\end{equation}
where $\mathbf{x}_0$ is sampled from training data and $\mathbf{x}_t \sim q(\mathbf{x}_t|\mathbf{x}_0)$. Once we have access to the well-trained $\epsilonb_\thetab(\x_t,t)$, a clean sample can be derived by evaluating the generative reverse process Eq.~\eqref{eq:ddpmreverse} step by step.

Furthermore, one can view the DDPM equivalent to the variance preserving (VP) form of the stochastic differential equation (SDE)~\cite{song2021scorebased}. 
Accordingly, the epsilon-matching objective Eq.~\eqref{eq:epsilon-matching} is equivalent to the denoising score-matching~\cite{sohl2015deep} objective with different parameterization:
\begin{equation}
\label{eq:score-matching}
\mathbf{L}({\thetab}) = \mathbb{E}_{\mathbf{x}_0, {\epsilonb}, t} [||\mathbf{s}_{{\thetab}} (\mathbf{x}_t, t) - \nabla_{\mathbf{x}_t} \log q(\mathbf{x}_t|\mathbf{x}_0) ||^2].
\end{equation}

\subsection{Diffusion Based Solvers for Image Restoration}
This paper focuses on solving the image restoration, or image inverse problems with unconditional diffusion model~\cite{kawar2022denoising,chung2023diffusion,song2023pseudoinverseguided}. Our goal is to retrieve the unknown $\x_0$ from a degraded measurement $\y$:
\begin{equation}
    \label{eq:ax+n}
    \mathbf{y} = \mathcal{A}(\mathbf{x}_0) + \mathbf{n},\quad \mathbf{y},\mathbf{n} \in \Rd^n,\, \mathbf{x} \in \Rd^{d},
\end{equation}
where $\mathcal{A}(\cdot): \Rd^d \mapsto \Rd^n$ is the degradation operator and $\mathbf{n} \sim \mathcal{N}(\mathbf{0},{\sigma_y^2}\Ib)$ denotes the measurement noise.
The restoration problem can be addressed via conditional diffusion models by substituting the score function $\nabla_{\mathbf{x}_t} \log (\mathbf{x}_t)$ in the reverse-time SDE with the conditional score function $\nabla_{\mathbf{x}_t} \log (\mathbf{x}_t|\mathbf{y})$, which can be derived by Bayes' rule:
\begin{equation}
    \label{eq:grad_log_bayes}
    \nabla_{\x_t} \log p_t(\x_t|\y) = \nabla_{\x_t} \log p_t(\x_t) + \nabla_{\x_t} \log p_t(\y|\x_t).
\end{equation}
The first prior term can be approximated by a well-trained score network. However, the analytic solution of the second likelihood term is computationally 
intractable, since there only exists an explicit connection between $\y$ and $\x_0$. 
To solve this dilemma, \citet{chung2023diffusion} propose diffusion posterior sampling (DPS) to approximate the likelihood using a Laplacian approximation:
$ \nabla_{\x_t} \log p(\y|\x_t) \simeq   \nabla_{\x_t} \log p (\y|\hat \x_{0|t}),$
where $\hat\x_{0|t}$ is the denoised estimate via Tweedie’s formula~\cite{stein1981estimation,efron2011tweedie}.
Consequently, the generative distribution $p_\thetab(\x_{t-1}| \\ \x_{t},\y)$
can be modeled as:     
\begin{align}   \label{eq:p(xt-1_xt_y)}
p_\thetab(\x_{t-1}|{\x}_{t},\y) : = \mathcal N \big(\x_{t};~{\mub}_{\thetab}(\x_{t},t,\y),~{\sigma_{t}^2}\Ib\big).
\end{align}
The mean of Gaussian distribution ${\mub}_{\thetab}(\x_{t},t,\y)$ is obtained by:
\begin{equation}
    \begin{aligned}       
{\mub}_{\thetab}(\x_{t},t,\y) = & \frac{1}{\sqrt{\alpha_t}} \Big( \x_{t} - \frac{\beta_t}{\sqrt{1-\bar{\alpha}_t}} {\epsilonb}_\thetab(\x_{t},t) \Big) \\
& -\lambda_t \nabla_{\x_t} ||\y-\mathcal{A} \hat \x_{0|t}||,
  \end{aligned}\label{eq:meanofdps}
\end{equation}  
where $\hat\x_{0|t}=\frac{1}{\sqrt{\bar{\alpha}_t}} \big(\x_t-\sqrt{1-\bar{\alpha}_t}{\epsilonb}_\thetab(\x_{t}, t)\big)$, and $\lambda_t$ is a tunable step size.

\section{Diffusion Posterior Proximal Sampling}

\subsection{Random Sampling Induces Uncertainties and Error Accumulation}
Existing approaches~\cite{song2022solving,kawar2022denoising,chung2022improving,wang2023zeroshot,song2023pseudoinverseguided} address image restoration or inverse problems following a process derived from unconditional generation. Specifically, the process (i) initiates the denoising process with pure white noise, and (ii) incorporates random noise at each reverse step.
We contend that randomness is unsuitable for restoration problems that demand the preservation of measurement identities, such as super-resolution or deblurring. Furthermore, the cumulative impact of random noise at each step results in the smoothing of output, thereby yielding low-quality generated samples.
Recent investigations~\cite{ma2023OptimalBoundary,everaert2023SignalLeak} align with our findings, highlighting that randomness introduces instability and fluctuations, ultimately culminating in suboptimal samples.

On the other hand, we observe a discrepancy in the inputs provided to the noise prediction network $\epsilonb_\thetab$. During the training phase, the network is fed with ground truth training samples. However, in the inference stage, the input $\x_t$ is randomly sampled from the predicted distribution. In cases where the predicted distribution is inaccurate, random sampling can exacerbate the deviation from the expected values, introducing substantial exposure bias~\cite{ning2023input} to the generative process.
While the data consistency update $\nabla_{\x_t} ||\y-\mathcal{A} \hat \x_{0|t}||$ does offer mitigation by adjusting the sample to align with the measurement, it demands delicate design and specifically-tuned parameters~\cite{song2023loss}. And the optimal parameter values vary across datasets and tasks. Consequently, the selection of sampling choices holds significance as it directly influences the input to the network, thereby affecting the generation quality.

\subsection{Proximal Sampling at Each Step}

To tackle the challenges posed by random sampling, in this paper, we propose to extract multiple candidate samples from the predicted distribution, and select the most proximal one~\cite{parikh2014proximal} to our anticipated target. Our motivation comes from the following idea:
considering $\x_0$ is an ideal but unknown solution for the image restoration problem, 
the sample taken from the predicted distribution $\x_{t-1}\sim p_\thetab(\x_{t-1}|\x_{t},\y) $ should be close to posterior $q(\x_{t-1}|\x_t,\x_{0})$, as it models the desired reverse process. 

The mathematical formulation of our selection process is given by:
\begin{align}
\label{eq:findz1}
    \x_{t-1} = &\argmin_{\x_{t-1}^i} \|\x_{t-1}^i -\x_{t-1}^* \|_2^2 
\end{align}
where $\x_{t-1}^i \sim p_\thetab(\x_{t-1}|\x_{t},\y),~i \in [0,n-1]$,  and $\x_{t-1}^*$ is our proposed deterministic solution~\footnote{the symbol $^*$ means an ideal but unknown solution.} via DDIM~\cite{song2020denoising} with unknown $\x_0$ (see supplementary material for detailed derivation):        
\begin{equation}
\begin{aligned}
    \x_{t-1}^*=&\sqrt{\Bar{\alpha}_{t-1}} {\x}_{0}+\sqrt{1-\Bar{\alpha}_{t-1}} \cdot \frac{\x_t-\sqrt{\Bar{\alpha}_{t}}\x_0}{\sqrt{1-\Bar{\alpha}_{t}}}\\
    =&C_1 \x_t+ C_2 \x_0.
\end{aligned}
\end{equation}
Here $C_1 = {\sqrt{1-\Bar{\alpha}_{t-1}}}/{\sqrt{1-\Bar{\alpha}_{t}}}$ and $C_2=\sqrt{\Bar{\alpha}_{t-1}} - {\sqrt{\Bar{\alpha}_{t}} \sqrt{1-\Bar{\alpha}_{t-1}}}/ \\{\sqrt{1-\Bar{\alpha}_{t}}}$ are introduced for simplicity.    

\begin{figure}[t]
    \centering
    \includegraphics[width=1.\linewidth]{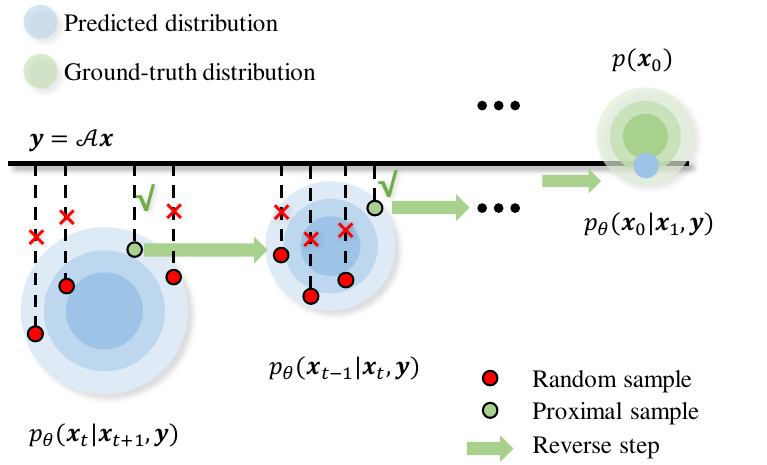}
    \caption{Conceptual illustration of our DPPS. Our method extract multiple candidate samples from the predicted distribution and choose the one with the highest measurement consistency. }
    \label{fig:manifold_anal} \Description{Conceptual illustration of our DPPS. Our method extract multiple candidate samples from the predicted distribution and choose the one with the highest measurement consistency.}
    \vspace{-0.3cm}
\end{figure}

The underlying motivation of our approach is depicted in \cref{fig:manifold_anal}. Without changing the diffusion denoiser,
our approach directs the sampling results toward a predefined target, mitigating the drawbacks induced by randomness. Moreover, the generative process remains stochastic, benefiting from the injected noise~\cite{bao2021analytic,karras2022elucidating} since it corrects prediction errors and imprecise parameter settings from previous steps. Consequently, our method converges more rapidly and attains superior results.

However, accessing $\x_0$ during inference is infeasible, rendering $\x_{t-1}^*$ also unknown.
Now, one important contribution of this paper is to project Eq.~\eqref{eq:findz1} onto the measurement subspace. This is feasible as $\Ac \x_{t-1}^*$ can be approximated under the condition $\y \approx \Ac \x_0$ when the degradation operator is linear and $\sigma_y$ is assumed within a moderate range.
\begin{equation}
\begin{aligned} \label{eq:Axtappr}
  \Ac \x_{t-1}^*  \approx C_1 \Ac \x_t+ C_2 \y. 
\end{aligned}
\end{equation}
Specifically, by projecting onto the measurement subspace, we choose $\x_{t-1}$ by the following:
\begin{align}
    \x_{t-1}  = \argmin_{\x_{t-1}^i} \|\Ac(\x_{t-1}^i -\x_{t-1}^*) \|_2^2
\end{align}
Since $\x_0$ is expected to be more accurate than the denoised result $\hat\x_{0|t}$, the measurement $\y \approx \Ac \x_0$ can provide a stable and strong supervisory signal to correct the sampling result.
Then by means of sample selection, we effectively control the result $\x_{t-1}$ within a more desirable region, even though we have no access to the $\x_0$.

As $\Ac$ is irreversible and complex in numerous scenarios~\cite{chen2021equivariant,chen2023imaging}, obtaining its pseudo-inverse poses challenges. Therefore, we propose to randomly sample $n$ candidates $\x_{t-1}^i, i \in [0,n-1]$ from the predicted distribution  $ p_\thetab(\x_{t-1}|\x_{t},\y)$, and choose the one with the minimal deviation from our anticipated values. Finally, Eq.~\eqref{eq:findz1} becomes:
\begin{equation}
\begin{aligned}
\label{eq:findz_f}
    \z_{t}  =& \argmin_{\z_{t}^i} \| \Ac\left(\mub_{\thetab}(\x_t,t,\y)+\sigma_t \z_t^i -C_1 \x_t \right) -C_2\y \|_2^2. 
\end{aligned}
\end{equation}
The process of sample selection is also denoted as noise selection, given by $\x_{t-1}^i = \mub_{\thetab}(\x_t,t,\y)+\sigma_t \z_t^i$.
We detail the full algorithm in \cref{alg:ours_alg}, and name our algorithm Diffusion Posterior Proximal Sampling (DPPS). 
The chosen $\x_{t-1}$ is referred to as the proximal sample due to the similarity to the proximal operator~\cite{parikh2014proximal}.

\subsection{Adaptive Sampling Frequency}
It is commonly perceived that the generative process of diffusion does not exhibit uniform significance across all timesteps~\cite{karras2022elucidating,kingma2023understanding,esser2024scaling}. Based on this insight, we adaptively adjust the number of candidate samples during the generative process according to the signal-to-noise ratio to enhance image quality and reduce computational cost. Through extensive experimentation, we discovered that using a larger number of candidate samples in the final stage of generation yields better results. The number of candidate samples $n$ at each step can be expressed as:
\begin{equation}
    n = \max (\lfloor N_{max} * (1-e^{-\lambda_t})\rfloor,~2) ,
\end{equation}
where $t \in [0, 999]$ is the timestep, $\lambda_t = \frac{\bar{\alpha}_t}{1-\bar{\alpha}_t}$ represents the signal-to-noise ratio, and $N_{max}$ is a hyper-parameter that adjusts the maximum sampling frequency. The minimum frequency is set to $2$, ensuring our proximal sampling strategy.

\subsection{Aligned Initialization}
Recent studies have pointed out that initializition have a significant impact on the generated results~\cite{singh2022conditioning,cao2023deep,ma2023OptimalBoundary}, and the discrepancy between the two distributions account for a discretization error~\cite{benton2023linear}.
In this paper, inspired by \cite{everaert2023SignalLeak}, we simply initialize the sample in the same way as during training, making the best use of the available measurement
\begin{equation}    
\begin{aligned}
    \x_T  = \sqrt{\bar{\alpha}_T} \Ac^T \y + \sqrt{1-\bar{\alpha}_T} \epsilonb \quad
 \epsilonb \sim \Nc(\mathbf{0},\Ib).
\end{aligned}
\end{equation}
where $\Ac^T$ means the transpose of operator. 
We argue this trick provides a modicum of information as a signal to the reverse model, as it realigns the distribution of initial latent with the training distribution~\cite{everaert2023SignalLeak}.

\begin{algorithm}[t]
    \caption{Diffusion Posterior Proximal Sampling}
    \label{alg:ours_alg}
    \begin{algorithmic}[1]
        \REQUIRE $T$, $\y$, \{$\lambda_t\}_{t=1}^T, {\{\sigma_t\}_{t=1}^T}$, $n$
        \STATE $\x_T = \sqrt{\bar{\alpha}_T} \Ac^T \y + \sqrt{1-\bar{\alpha}_T} \epsilonb, \quad  \epsilonb \sim \Nc(\mathbf{0},\Ib)$
        \FOR{$t=T$ {\bfseries to} $1$}
            \STATE compute $\mub_{\thetab}(\x_t,t,\y)$ via \eqref{eq:meanofdps}
            \FOR{$i=n$ to $1$}
            \STATE $ {\z_t^i} \sim\Nc(\mathbf{0},\Ib)$
            \STATE $D_i$ = $\| \Ac\big(\mub_{\thetab}(\x_t,t,\y) + \sigma_t {\z_t^i}   - C_1 \x_t \big) -  C_2 \y \|_2^2 $       \ENDFOR
            \STATE $\z_t = \argmin_{\z_t^i} D_i$
            \STATE {$\x_{t-1} \gets \mub_{\thetab}(\x_t,t,\y) + \sigma_t \z_t $ }
        \ENDFOR        
        \STATE \textbf{return} $\hat{\x}_0$
    \end{algorithmic}
\end{algorithm}

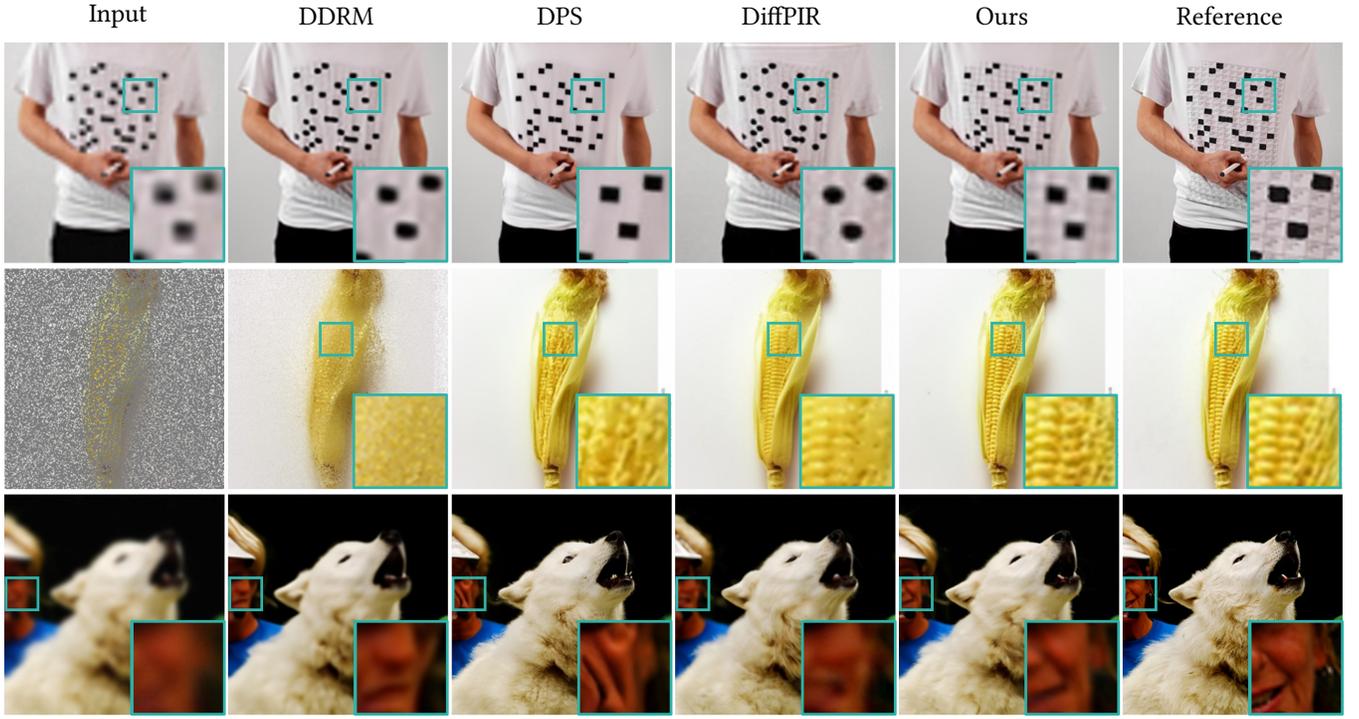
\begin{figure*}[ht]
    \centering
    \begin{tikzpicture}
        \node[anchor=south west,inner sep=0] (image) at (0,0) {\includegraphics[width=\textwidth]{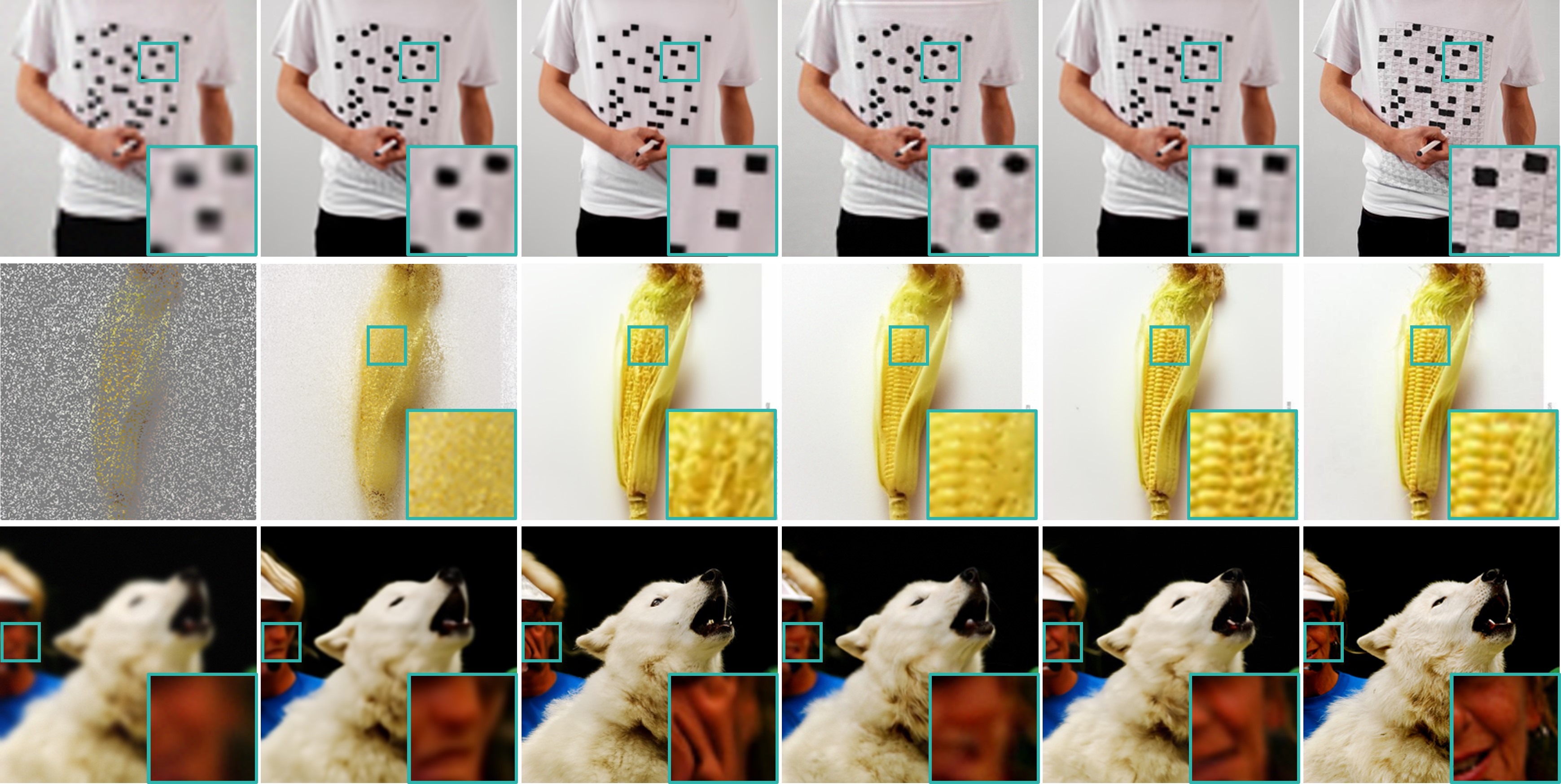}};
        \begin{scope}[x={(image.south east)},y={(image.north west)}]
            \node at (0.085,1.04) {\large{Input}};
            \node at (0.248,1.04) {\large{DDRM}};
            \node at (0.415,1.04) {\large{DPS}};
            \node at (0.58,1.04) {\large{DiffPIR}};
            \node at (0.745,1.04) {\large{Ours}};
            \node at (0.915,1.04) {\large{Reference}};
        \end{scope}
    \end{tikzpicture}
    \caption{Image restoration results with $\sigma_y=0.01$. Row 1: SR ($\times$4), Row 2: 80\% inpainting, Row 3: Gaussian deblurring.}
    \label{fig:results_linear} \Description{Image restoration results with $\sigma_y=0.01$. Row 1: SR ($\times$4), Row 2: 80\% inpainting, Row 3: Gaussian deblurring.}
    \vspace{-0.2cm}
\end{figure*}

\subsection{Discussion}
 
\subsubsection{Relevance to Existing Methods.}
DPS~\cite{chung2023diffusion} proposes to estimate the intractable likelihood under Laplacian approximation, and adopts random sampling from the predicted distribution. However, the randomly chosen sample may not be well-coordinated with the measurement information.
Another study MCG~\cite{chung2022improving} iteratively applies projection onto the measurement subspace after the denoising step to ensure data consistency.
However, the imposed projections lead to accumulated error due to measurement noise.

Our algorithm can be viewed as a special case of DPS, where we choose the sample that exhibits better measurement consistency. It absorbs the advantages of the robustness to noise from DPS and the faithful data consistency of MCG.
Moreover, the better data consistency is achieved by injecting a guided 
adaption $\z_t$, which also helps correct the prediction errors in last generation steps~\cite{karras2022elucidating,jolicoeur2021gotta}, leading to enhanced output quality.

\subsubsection{Computational Efficiency.}
Selecting the proximal sample introduces some extra computational overhead. In practice, the diffusion model conducts back-propagation only once, aligning with the mainstream  approaches~\cite{song2022solving,chung2022improving,chung2023diffusion,song2023loss}. The evaluations for Eq.~\eqref{eq:findz_f} are considerably cheaper than gradient calculation, making the additional computational costs manageable.
We show the computational resources for different settings in \cref{sec:Compute_time}.

\subsubsection{Theoretical Analysis.}
Here, we provide some theoretical supports for our methodology. Despite our proposed proximal sampling being random, we can theoretically approach the desired point within an upper bound with mild assumptions.
The proposition, detailed in supplementary material, proves that we can converge to our desired target $\x_{t-1}^*$, when the number of candidate samples is large enough.
Furthermore, our method doesn't require an exceptionally large number of candidates to achieve promising performance. This is because the variance of our method decreases when the proximal one is selected from any number of candidates.

\begin{proposition}  \label{prop:var}
    For the random variable $\z \sim \mathcal{N}(\mathbf{0},\Ib)$ and its objective function:
    \begin{equation}
        f(\z) = \|\Ac(\mub_{\thetab}(\x_t,t,\y)+\sigma_t \z 
        -C_1 \x_t ) -C_2\y \|^2_2.
    \end{equation}
    The variance for $f(\z)$ is denoted by $\text{Var}\left(f(\z)\right)$. We have
    \begin{equation}
        \text{Var}_\text{DPS}\left(f(\z)\right) > \text{Var}_\text{MC}\left(f(\z)\right) > \text{Var}_\text{Ours}\left(f(\z)\right).
    \end{equation}
    Here, $\text{Var}_\text{DPS}\left(f(\z)\right)$ is the variance of DPS~\cite{chung2023diffusion}, $\text{Var}_\text{MC}\left(f(\z)\right)$ is the variance of Monte Carlo sampling, and $\text{Var}_\text{Ours}\left(f(\z)\right)$ is the variance of our proximal sampling method.
\end{proposition}
In line with our theoretical analysis, empirical observations suggest that the proximal sample strategy demonstrates improved data consistency and reduced stochasticity, as evidenced by the experimental results.

\begin{table*}[t]
\centering 
\caption{\centering
Quantitative evaluation of image restoration tasks on FFHQ 256$\times$256-1k and ImageNet 256$\times$256-1k with $\sigma_y=0.01$. \textbf{Bold}: best, \underline{underline}: second best.
}
\vspace{-0.3cm}
\resizebox{1.0\textwidth}{!}{
\begin{tabular}{lcccccccccccc}
\toprule
{\textbf{}} & \multicolumn{3}{c}{\textbf{Inpaint (random)}} & \multicolumn{3}{c}{\textbf{Deblur (Gaussian)}} &
\multicolumn{3}{c}{\textbf{Deblur (motion)}} &\multicolumn{3}{c}{\textbf{SR ($\times$ 4)}} \\
\cmidrule(lr){2-4}
\cmidrule(lr){5-7}
\cmidrule(lr){8-10}
\cmidrule(lr){11-13}
{\textbf{Method}} & {PSNR $\uparrow$~} &  {LPIPS $\downarrow$~} & {FID $\downarrow$~}& {PSNR $\uparrow$~} &  {LPIPS $\downarrow$~} & {FID $\downarrow$~} & {PSNR $\uparrow$~} &  {LPIPS $\downarrow$~} & {FID $\downarrow$~} & {PSNR $\uparrow$~} &  {LPIPS $\downarrow$~} & {FID $\downarrow$} \\
\toprule
\multicolumn{13}{c}{\textbf{FFHQ}} \\
\toprule
PnP-ADMM~\cite{chan2016plug} & 27.32 & 0.349 & 63.19 & 25.97  & 0.260 & 94.50 & \underline{25.86} & 0.278 & 59.06 & \underline{26.94} & 0.292 & 90.11 \\
Score-SDE~\cite{song2021scorebased}~ & 21.91  & 0.371 & 58.19 & 18.37 & 0.629 & 169.68 & 15.79  & 0.635 & 176.73 & 24.10 & 0.363 & 75.50 \\
MCG~\cite{chung2022score}& 24.59 & 0.265 & 29.31  & 19.67 & 0.497 & 85.85 & 18.00 & 0.604 & 91.73  & 24.02 & 0.321 & 52.88 \\
DDRM~\cite{kawar2022denoising} & 21.85 & 0.378 & 78.40 & \underline{26.38} & 0.298 & 62.72 & - & - & - & \textbf{27.31} & 0.248 & 52.49 \\
DPS~\cite{chung2023diffusion} & 27.29   & \underline{0.182} & \underline{23.11}  & 26.04 & \underline{0.228} & \underline{30.52} & 25.24 & \underline{0.261} & \underline{33.47} & 26.55   & \underline{0.237} & 34.50 \\
LGD-MC~\cite{song2023loss} & 26.11 & 0.257 & 35.26 & 24.08 & 0.288 & 34.54 & 21.04 & 0.370 & 49.19 & 26.12 & 0.247  & 33.68 \\
DiffPIR~\cite{zhu2023denoising} & \textbf{27.96} & 0.210 & 30.38 & 25.38 & 0.276 & 32.00 & 22.74 & 0.331 & 83.42 & 24.74 & 0.273 & \underline{33.03} \\
\midrule
\rowcolor{lightblue}
Ours & \underline{27.50} & \textbf{0.161}  & \textbf{18.65} & \textbf{26.42} & \textbf{0.221} & \textbf{28.93} & \textbf{27.82} & \textbf{0.197} & \textbf{28.28} & 26.94 & \textbf{0.201}  & \textbf{25.98}  \\
\bottomrule
\multicolumn{13}{c}{\textbf{ImageNet}} \\
\toprule
PnP-ADMM~\citep{chan2016plug} & 25.14 & 0.405 & 66.54  & 21.97 & 0.419 & 63.15 & \textbf{21.86} & \underline{0.408} & \underline{61.46} & 23.95 & 0.346 & 71.41 \\
Score-SDE\cite{song2021scorebased}~ & 16.25  & 0.653 & 102.56 & 21.31 & 0.467 & 82.54 & 13.56 & 0.661 & 89.62 & 17.69 & 0.624 & 96.67 \\
MCG~\cite{chung2022score}& 23.21 & 0.324 & 44.09 &  12.31 & 0.647 & 109.45 & 18.32 & 0.633& 99.26 & 17.08 & 0.538 & 85.91  \\
DDRM~\cite{kawar2022denoising} & 19.34 & 0.555 & 147.00 & \textbf{23.67} & 0.401 & 66.99 & - & - & - & \textbf{25.49}  & 0.319 & 54.77 \\
DPS~\citep{chung2023diffusion} & \underline{25.65}  & 0.240 & \underline{29.04}  & 19.65 & 0.422 & 65.35 & 18.79 & 0.458 & 77.29 & 23.88    & 0.335 & 42.83 \\
LGD-MC~\cite{song2023loss} & 24.06 & 0.316 & 40.95 & 20.32 &0.423 &62.79 & 19.07&0.461 & 78.79 & 22.78& 0.390 & 59.61\\
DiffPIR~\cite{zhu2023denoising} & \textbf{25.85}  & \underline{0.235} & 33.16 & 22.03  & \underline{0.395} & \underline{54.71} & 19.86 & 0.433 & 79.23 & \underline{24.78} & \underline{0.302} & \underline{39.25} \\
\midrule
\rowcolor{lightblue}
Ours & 24.97 & \textbf{0.217}  & \textbf{24.90} & \underline{22.70}  & \textbf{0.364} & \textbf{51.21} & \underline{21.65} & \textbf{0.375} & \textbf{51.35} & 24.44 & \textbf{0.267}  & \textbf{30.70}  \\

\bottomrule
\end{tabular}
}
\label{tab:results_all}
\vspace{-0.2cm}
\end{table*}

\section{Related Work}

\subsection{Diffusion Models for Image Restoration}
To solve image restoration or image inverse problems with diffusion models, plenty of notable works~\cite{dhariwal2021diffusion,kawar2022denoising,chung2023diffusion,song2023pseudoinverseguided} have been introduced. The first category of approaches~\cite{saharia2021image,dhariwal2021diffusion,saharia2022palette} involves training conditional diffusion models or approximating the likelihood with synthetic image pairs as training data. However, these methods require specific training for each task and lack generalization across a wide range of inverse problems.
In contrast, the second category of approaches addresses inverse problems with unconditional diffusion models by guiding the reverse diffusion process~\cite{song2022solving,zhu2023denoising}. Several studies~\cite{song2023pseudoinverseguided,chung2023diffusion,song2023loss} concentrate on estimating the time-dependent likelihood $p_t(\y|\x_t)$ for posterior sampling. Among them, \citet{song2023loss} enhances the likelihood approximation with a Monte-Carlo estimate. \citet{ma2023OptimalBoundary} tackles image super-resolution by selecting the best starting point. 
\cite{jin2024des3,jiang2023low,hou2024global} effectively constrain the results of each generative step, adeptly tackling low-light image enhancement and shadow removal.

Recently, latent diffusion models have also seen advancements~\cite{rombach2022high,vahdat2021score} and have been widely adopted in various image inverse problem scenarios~\cite{rout2023solving,he2023iterative,song2023SolvingIP,chung2023PrompttuningLD,rout2023BeyondFT}. Furthermore, while several deep learning-based proximal optimization algorithms have been proposed~\cite{wei2022tfpnp,lai2023prox}, our method stands out as the first to make the diffusion sampling as a proximal optimization process.

\subsection{The Exposure Bias in Diffusion Models}

The exposure bias~\cite{ning2023input,li2023alleviating} in diffusion models is described as the misalignment between the training input $\x_t$ and the reverse process input $\x_t$, which is essentially a mismatch between the predicted noise $\epsilonb_\thetab$ and the actual noise $\epsilonb$. \citet{ning2023input} mitigate the exposure bias problem by perturbing the training input, rendering the network more robust to inaccurate inputs during generation. \citet{ning2023elucidating} propose dynamic scaling to correct the magnitude error of $\epsilonb_\thetab$ and improve sample quality. \citet{li2023alleviating} present a novel solution by adjusting the corruption level of the current samples.
This paper provides a new way to effectively alleviate the exposure bias problem to some extent by reducing the uncertainty of random sampling.

\section{Experiments}
\subsection{Setup}
\paragraph{Dataset, Model.}
To showcase the effectiveness of the proposed methods, we conducted experiments on two standard datasets, namely FFHQ 256$\times$256~\cite{karras2019style} and ImageNet 256$\times$256~\cite{russakovsky2015imagenet}.  The evaluation encompasses the first 1k images in the validation set of each dataset. The diffusion models pre-trained on ImageNet and FFHQ are sourced from \citet{dhariwal2021diffusion} and \citet{chung2023diffusion}, respectively.
For comprehensive comparisons, we include state-of-the-art diffusion-based image restoration solvers, including the Plug-and-play alternating direction method of multipliers (PnP-ADMM)~\cite{chan2016plug}, Score-SDE~\cite{song2022solving}, Denoising Diffusion Restoration Models (DDRM)~\cite{kawar2022denoising}, Manifold Constrained Gradient (MCG)~\cite{chung2022improving}, DPS~\cite{chung2023diffusion}, LGD-MC~\cite{song2023loss}, and DiffPIR~\cite{zhu2023denoising}. Sampling frequency parameter $N_{max}$ is set to 50. To ensure fair comparisons, we utilized the same pre-trained diffusion models for all diffusion-based methods. 
All experiments were executed with a fixed random seed.

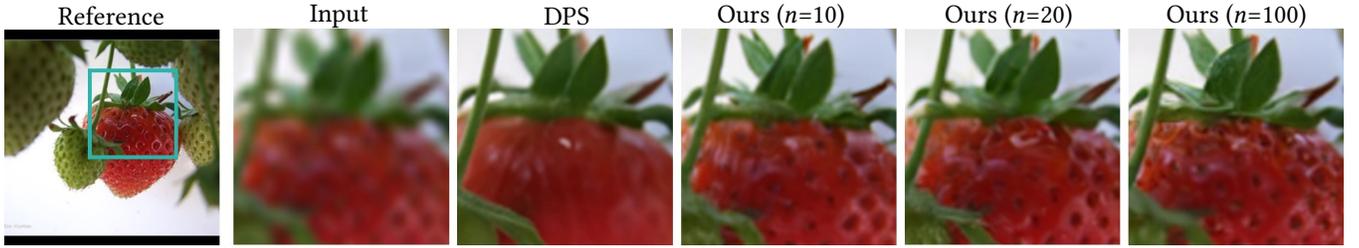
\begin{figure*}[t]
    \centering
    \begin{tikzpicture}
        \node[anchor=south west,inner sep=0] (image) at (0,0) {\includegraphics[width=\textwidth]{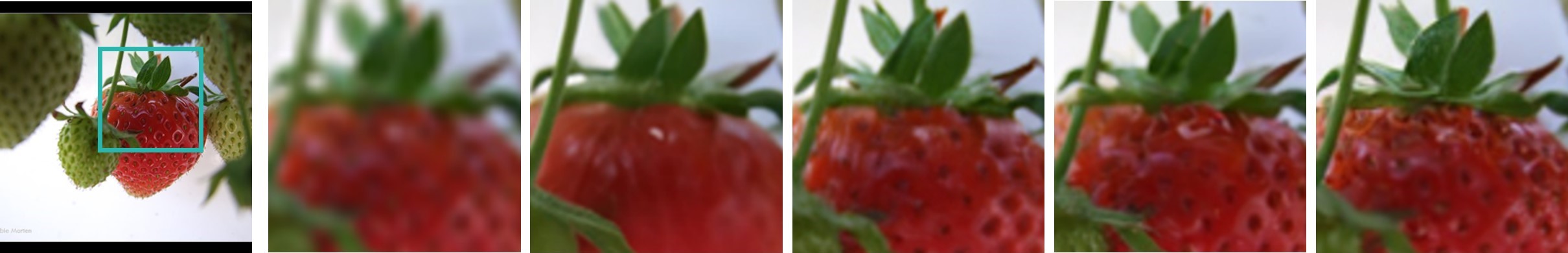}};
        \begin{scope}[x={(image.south east)},y={(image.north west)}]
            \node at (0.08,1.06) {\large{Reference}};
            \node at (0.25,1.06) {\large{Input}};
            \node at (0.42,1.06) {\large{DPS}};
            \node at (0.58,1.06) {\large{Ours ($n$=10)}};
            \node at (0.75,1.06) {\large{Ours ($n$=20)}};
            \node at (0.92,1.06) {\large{Ours ($n$=100)}};
        \end{scope}
    \end{tikzpicture}
    \caption{Visual results on SR ($\times$4) task to demonstrate the efficacy of our proposed method and to explore the impact of $n$.} \Description{Visual results on SR ($\times$4) task to demonstrate the efficacy of our proposed method and to explore the impact of $n$.}
    \label{fig:cover} 
    \vspace{-0.3cm}
\end{figure*}

\paragraph{Degradation Operator.}
The degradation operators are defined as follows:
(i) For inpainting, 80\% of the pixels (all RGB channels) in the image are masked. (ii) For Gaussian blur, a blur kernel of size $61\times61$ with a standard deviation of $3.0$ is employed. (iii) For motion blur, the kernel, following the procedure outlined in \citet{chung2022improving}, has a size of $61\times61$ and an intensity value of $0.5$. (iv) For SR $\times~4$, the operator involves $4\times$ bicubic down-sampling.
Additive Gaussian noise with a variance of $\sigma_y=0.01$ is applied for all degradation. 

\paragraph{Metrics.}
The metrics employed for the comparison encompass: Peak Signal-to-Noise Ratio (PSNR) as distortion metrics;  Learned Perceptual Image Patch Similarity (LPIPS)~\cite{Zhang2018TheUE} distance, and Frechet Inception Distance (FID)~\cite{Heusel2017GANsTB} as perceptual metrics.

\subsection{Main Results}

We present the statistic results of the general image restoration tasks on both FFHQ and ImageNet datasets, as detailed in \cref{tab:results_all}.
To the dataset with homogeneous scenarios, i.e., FFHQ, our proposed method demonstrates impressive results across all metrics, establishing its superiority over existing state-of-the-art methods.
When dealing with more varied images within the ImageNet dataset, our approach exhibits substantial outperformance against all baseline methods in terms of FID and LPIPS, while maintaining comparable levels in PSNR.

The visual results for inpainting, SR, and Gaussian deblurring are shown in \cref{fig:results_linear}, showcasing the evident superiority of our proposed method.
While DDRM attains commendable distortion results with fewer Neural Function Evaluations (NFEs), 
it faces limitations in reliably restoration results for image inpainting tasks characterized by a very low rank of the measurement.
DiffPIR achieves satisfactory results in various scenarios; however, its performance is tied to the analytical solution for the data consistency term and exhibits sensitivity to measurement noise.
DPS differs from our approach in that it introduces random noise at each generation step, leading to over-smoothed and unstable restoration results, as depicted in \cref{fig:results_linear} (column 3). Conversely, our proposed method circumvents such drawbacks. The generative process is stabilized by directing the sample to a predefined target, yielding realistic and detailed restoration outcomes.
It is noteworthy that the generated results by our method exhibit minimal variation with different seeds (refer to supplementary material), aligning with the intended design of identity preserving.

\subsection{The Effect of Proximal Sampling} \label{sec:result_selection}

\begin{figure}[t]
    \centering  
    \includegraphics[width=1.\linewidth]{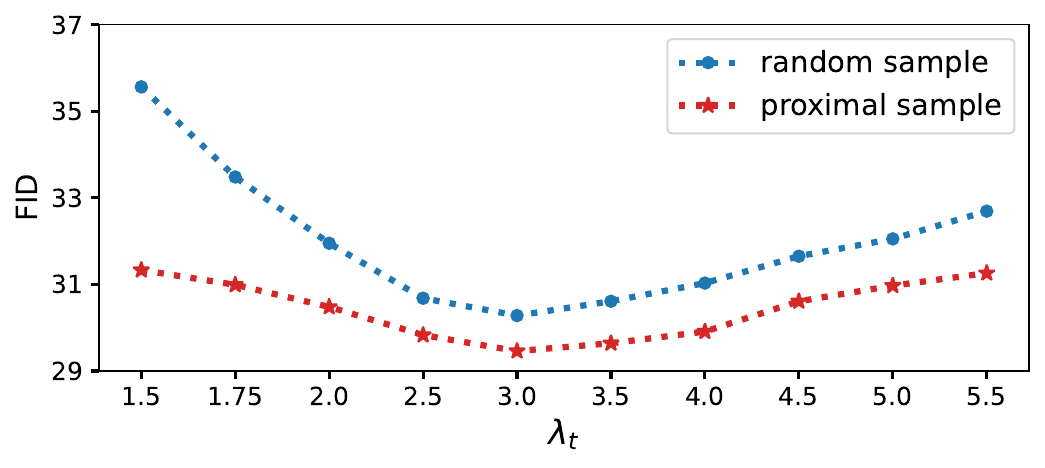}
    \caption{FID comparison between random sampling and proximal sampling with different values of $\lambda_t$. } \Description{FID comparison between random sampling and proximal sampling with different values of $\lambda_t$.}
    \label{fig:fid_lambda}
\vspace{-0.3cm}
\end{figure}

\begin{table}[t]
\centering
\caption{Comparision of computational resources and performance with different $n$ on FFHQ SR$\times$ 4.}
\resizebox{1. \linewidth}{!}{%
\begin{tabular}{lccc}
\toprule
Setting  & Memory / Growth  & Speed / Growth  & LPIPS / Gain   \\
\midrule
{$n$=1}  & 2599~MB /~ 0.0~\%  &56.36~s /~ 0.0~\% & 0.235~/~ 0.0~\% \\
{$n$=2}  & 2603~MB /~ 0.2~\%  &56.64~s /~ 0.5~\% &  0.221~/~ 6.0~\% \\
{$n$=10}  & 2621~MB /~0.8~\%  &56.82~s /~ 0.8~\% &  0.207~/~ 11.9~\%\\
{$n$=20}  & 2643~MB /~ 1.7~\% &57.25~s   /~ 1.5~\% & 0.203~/~ 13.6~\% \\
\rowcolor{lightblue}
adaptive $n$ & 2714~MB /~ 4.4~\% &56.86~s   /~ 0.9~\% & 0.201~/~ ~14.5\% \\
\bottomrule
\end{tabular}}
\label{tab:time_memory}
\vspace{-0.3cm}
\end{table}

\begin{figure*}[t] 
\centering 
\subfigure[MSE]{\includegraphics[width=0.33\textwidth]{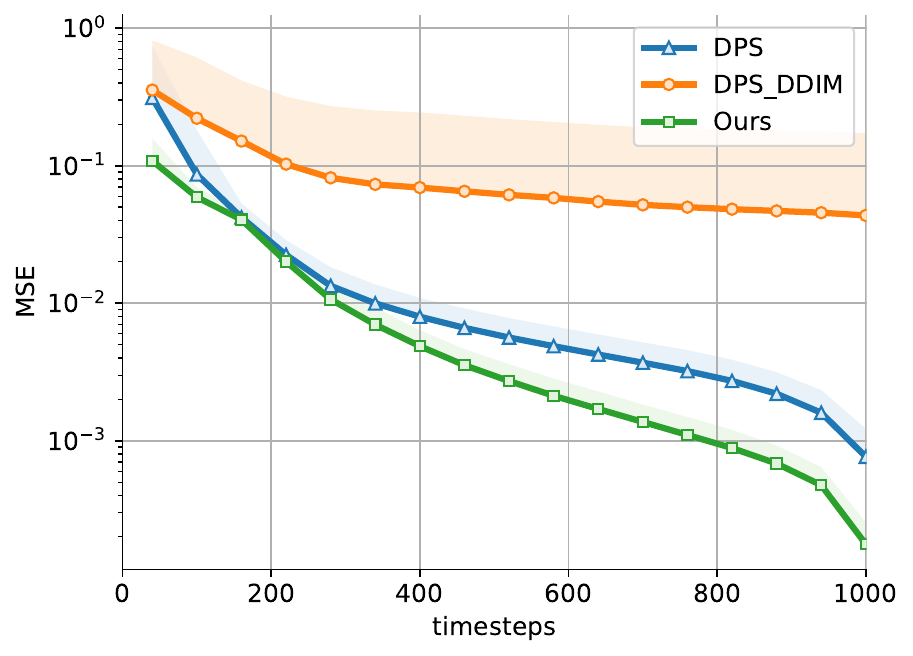}}
\subfigure[LPIPS]{\includegraphics[width=0.33\textwidth]{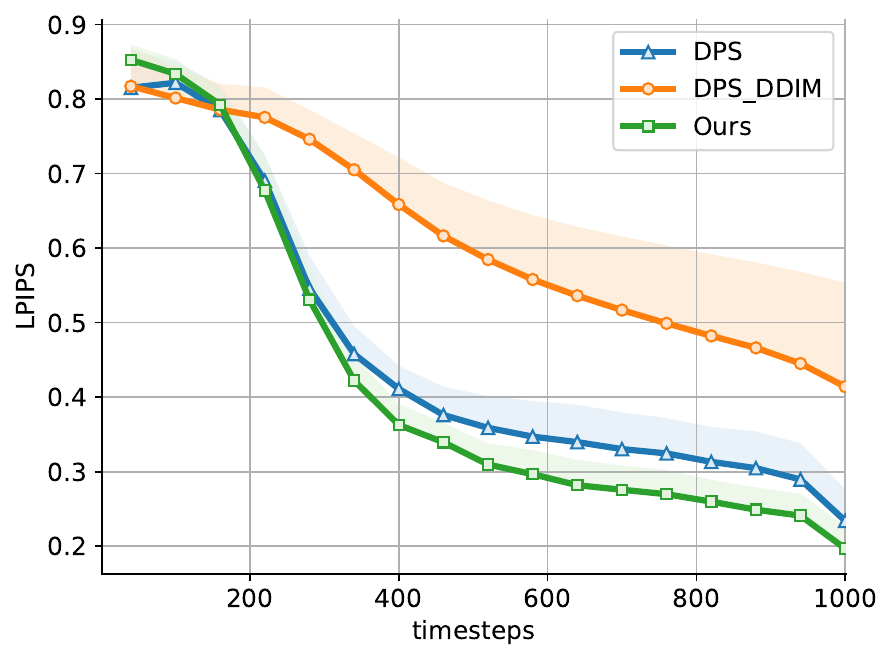}}
\subfigure[PSNR]{\includegraphics[width=0.33\textwidth]{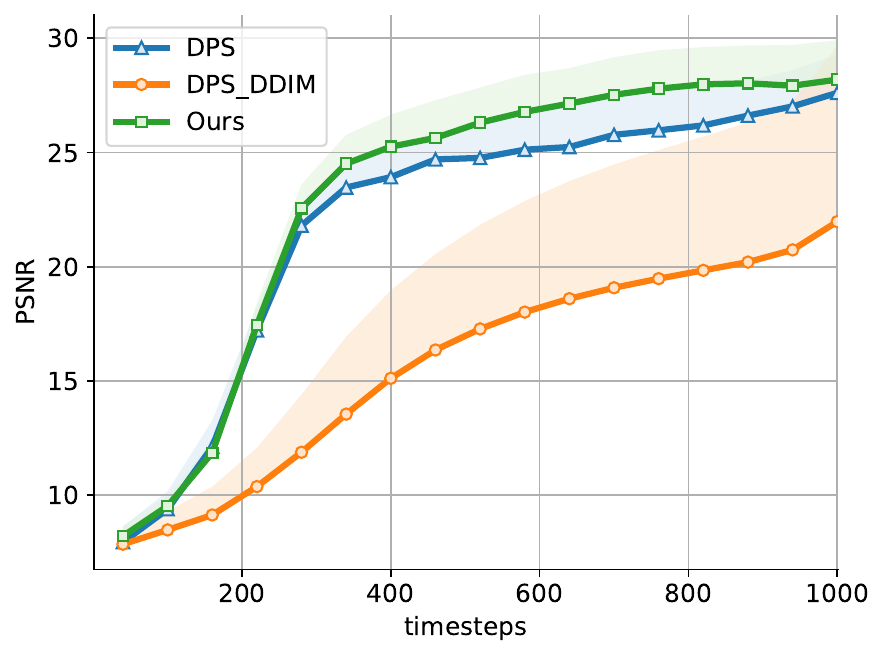}}
\caption{Convergence speed analysis. (a) mean error $\|\y-\Ac\hat{\x}_{0|t}\|^2_2$, (b) average LPIPS of $\hat{\x}_{0|t}$, and (c) average PSNR of $\hat{\x}_{0|t}$ with timesteps respectively. Our method achieves a faster optimization process and better restoration performance.} \Description{Convergence speed analysis. (a) mean error $\|\y-\Ac\hat{\x}_{0|t}\|^2_2$, (b) average LPIPS of $\hat{\x}_{0|t}$, and (c) average PSNR of $\hat{\x}_{0|t}$ with timesteps respectively. Our method achieves a faster optimization process and better restoration performance.}
\label{fig:Convergence_speed} 
\vspace{-0.3cm}
\end{figure*}

\subsubsection{Faster Convergence.}To investigate the impact of our proximal sampling on the generated results, we conducted further experimental analyses using the first 150 images from the FFHQ validation set. In addition to the naive DPS, we incorporated the results of DPS with DDIM~\cite{song2020denoising} deterministic sampling, denoted as DPS\_DDIM, which mitigates the effects of randomness. Results for (a) mean of $\|\y-\Ac\hat{\x}_{0|t}\|_2^2$, (b) average LPIPS of $\hat{\x}_{0|t}$, and (c) average PSNR of $\hat{\x}_{0|t}$ over timesteps are reported in \cref{fig:Convergence_speed}. It is evident that our method facilitates a more stable optimization process, yielding markedly superior results compared to DPS and DPS\_DDIM within the same period. This observation also supports the claim that the measurement can provide reliable supervisory signals for the sampling result.

\subsubsection{More Robust to Hyper-parameter.} We anticipate that our proposed proximal sampling serves as an adaptive correction for prediction errors resulting from imprecise parameter settings. To substantiate this claim,
we conducted experiments with different $\lambda_t$, comparing the results of our proximal sampling with those obtained through random sampling.
As illustrated in \cref{fig:fid_lambda}, our method consistently outperforms random sampling across different parameter settings. Notably, the performance of the random sampling method exhibits fluctuations with changes in parameters. In contrast, our method demonstrates stability over a range of parameter configurations, thereby validating the assertion.

\subsection{Computational Resource}\label{sec:Compute_time}
To examine the impact of the variable $n$ on computational efficiency, we conducted an evaluation that consisted of measuring memory consumption and computational time required for the SR$\times 4$ task.
This evaluation was performed on a single NVIDIA RTX 3090 GPU, and the results are presented in Table \ref{tab:time_memory}.
Notably, when the value of $n$ is relatively small ($\leq$20), our method shows a substantial enhancement in the restoration metric as $n$ increases while incurring minimal increases in memory consumption and computation time. 
Additionally, the carefully designed adaptive sampling frequency strategy further enhances restoration efficiency within less computation time.
These findings highlight the effectiveness and validity of our approach.

\begin{table}[t] 
\caption{Ablation studies on the number of candidate samples on FFHQ SR$\times$ 4.}  
        \centering
        \resizebox{1.\linewidth}{!}{%
        \begin{tabular}{lcccccc}
            \toprule        
            {} & \multicolumn{3}{c}{FFHQ}& \multicolumn{3}{c}{ImageNet} \\
            \cmidrule(lr){2-4}
            \cmidrule(lr){5-7}
            {Method}  & {PSNR$\uparrow$} & {LPIPS$\downarrow$} &{FID$\downarrow$}  &{PSNR$\uparrow$}& {LPIPS$\downarrow$} & {FID$\downarrow$}  \\
            \midrule
            $n$=1 & 26.44 & 0.235 & 34.91 & 23.95 & 0.336 & 43.54 \\
		$n$=2 & 26.64 & 0.221 & 31.19 & 24.42& 0.328 & 40.63  \\
		$n$=10 & 26.77 & 0.207 & 27.60 & \underline{24.64}& 0.282 & 34.05  \\   
		 $n$=20 & \underline{26.93}  & \underline{0.203} & \underline{26.30} & \textbf{24.79} & \underline{0.276} & \underline{33.75}  \\       
         \rowcolor{lightblue}
		adaptive $n$ & \textbf{26.94}  & \textbf{0.201} & \textbf{25.98} &  24.44 & \textbf{0.267} & \textbf{30.70} \\
            \bottomrule
        \end{tabular}
        }
        \label{tab:ablation_MC_times}
        \vspace{-0.2cm}
\end{table}

\subsection{Ablation Study}

\subsubsection{The Number of Candidate Samples.} \label{para:MC_time}
The number of samples, denoted as $n$, stands out as a pivotal parameter of the proposed methodology. The ablation studies relating to $n$ are shown in \cref{tab:ablation_MC_times} and \cref{fig:cover}.
\cref{fig:cover} illustrates a noticeable enhancement in output quality and details with increasing $n$.
This observation is reinforced by \cref{tab:ablation_MC_times}, where both the LPIPS and FID metrics exhibit a notable decrease as $n$ increases.
The augmented number of candidate samples corresponds to a higher likelihood of proximity to the intended target, thereby providing enhanced supervision for the generative process. 
The enhancement of the proposed adaptive sampling frequency further 
enhance the efficiency and applicability of our proposed method.

It is noteworthy that the improvement becomes relatively marginal when $n$ exceeds 20 and may even lead to a decrease in PSNR. There are several reasons for this phenomenon. (i) The approximation is only conducted within the range space, neglecting the null space. This can lead to a certain degree of misalignment with the reference image. (ii) Since we utilized $\y \approx \Ac \x_0$ in our algorithm, the presence of measurement noise ($\mathbf{n}$ in Eq.~\eqref{eq:ax+n}) affects the accuracy of our sample selection. (iii) It can also be attributed to the well-known trade-off between perception and distortion metrics~\cite{blau2018perception}. Consequently, a larger value of $n$ does not necessarily lead to better results.

\subsubsection{Noise Level.}
The precision of the approximation is impacted by the level of measurement noise $\mathbf{n}$. 
We explore the effects of different noise levels on the proposed method compared with random sampling.
The results presented in \cref{tab:ablation_of_noise_level} reveal that (i) both methods experience a decline in performance as the noise level increases, and our approach consistently outperforms random sampling in all settings. (ii) For inpainting, both methods experience considerable degradations as the noise level increases, demonstrating sensitivity to the noise levels.  (iii) For SR, notable performance degradation is observed in our method as the noise level increases, while the method adopting random sampling demonstrates greater robustness in this regard.

\begin{table}[t] 
\caption{Comparison of the impact of different noise levels.}
        \centering        
        \resizebox{1.\linewidth}{!}{%
        \begin{tabular}{cccccccc}
            \toprule
             & & \multicolumn{3}{c}{Inpainting} & \multicolumn{3}{c}{SR$\times$4} \\
            \cmidrule(lr){3-5}
            \cmidrule(lr){6-8} 
            $\sigma_y$ & Method & {PSNR$\uparrow$} & {LPIPS$\downarrow$} &{FID$\downarrow$}  &{PSNR$\uparrow$}& {LPIPS$\downarrow$} & {FID$\downarrow$} \\ 
            \midrule
             & Random & \textbf{27.69} & 0.173& 21.11 & 26.61&0.235&34.55 \\
		  \multirow{-2}{*}{0.0}&   \cellcolor{lightblue}{Ours} & \cellcolor{lightblue}27.61 & \cellcolor{lightblue}\textbf{0.156}& \cellcolor{lightblue}\textbf{17.50} & \cellcolor{lightblue}\textbf{27.42} & \cellcolor{lightblue}\textbf{0.187} & \cellcolor{lightblue}\textbf{24.08}  \\
		\midrule
		 & Random & 27.29 & 0.182& 23.11 & 26.55&0.237&34.50 \\
		 \multirow{-2}{*}{0.01} & \cellcolor{lightblue}{Ours} & \cellcolor{lightblue}\textbf{27.50} & \cellcolor{lightblue}\textbf{0.161} & \cellcolor{lightblue}\textbf{18.65} & \cellcolor{lightblue}\textbf{26.94} & \cellcolor{lightblue}\textbf{0.201} & \cellcolor{lightblue}\textbf{25.98} \\
		\midrule
		 & Random & \textbf{26.68} & 0.228& 30.63 & 25.89 & 0.257 & 34.59 \\
		\multirow{-2}{*}{0.05} &  \cellcolor{lightblue}{Ours} & \cellcolor{lightblue}26.51 & \cellcolor{lightblue}\textbf{0.208} & \cellcolor{lightblue}\textbf{24.62} & \cellcolor{lightblue}\textbf{26.07}& \cellcolor{lightblue}\textbf{0.245}& \cellcolor{lightblue}\textbf{30.67} \\
            \bottomrule
        \end{tabular}
        }
\label{tab:ablation_of_noise_level}
\vspace{-0.2cm}
\end{table}

\section{Conclusion}
In this paper, we introduce a novel approach to mitigate problems caused by misaligned random sampling in diffusion-based image restoration methods. 
Specifically, our approach advocates selecting the proximal sample that is more consistent with the observed measurement in the predicted distribution. An adaptive sampling frequency strategy is followed to optimize the computational efficiency of the proposed method.
In addition, we propose a realignment of the inference initialization involving measurement information to better align with expected generations. Experimental results validate a substantial performance improvement compared to SOTA methods.
Our method innovatively takes advantage of stochastic sampling in the diffusion generative process, exploiting the sampling selection to enhance the generation quality and offering new insights for future diffusion inference algorithms.

\clearpage

\begin{acks}
This work was supported by the Fundamental Research Funds for the Central Universities (No.~1082204112364), the National Nature Science Foundation of China (No.~62106161), and the Sichuan University Luzhou Municipal Government Strategic Cooperation Project (No.~2022CDLZ-8). We thank the anonymous reviewers for their insightful comments and feedback, and extend our gratitude to Mengxi Xie from SJTU and Linrui Dai from CQU for their constructive suggestions on the figures. 
\end{acks}

\bibliographystyle{ACM-Reference-Format}
\balance
\bibliography{bib}


\clearpage

\input{appendix}

\end{document}

%% file: preamble.tex
\usepackage{pgfplots}
\pgfplotsset{compat=1.18}
\usepackage{xkcdcolors}

\definecolor{BrickRed}{rgb}{0.6,0,0}
\definecolor{RoyalBlue}{rgb}{0,0,0.8}
\definecolor{Tdgreen}{rgb}{0,0.4,0.7}
\definecolor{pinegreen}{rgb}{0.0, 0.47, 0.44}
\definecolor{cornellred}{rgb}{0.7, 0.11, 0.11}
\definecolor{cadmiumgreen}{rgb}{0.0, 0.42, 0.24}
\definecolor{spirodiscoball}{rgb}{0.06, 0.75, 0.99}

\usepackage{pifont}

\usepackage{tikz}
\usetikzlibrary{decorations.pathreplacing,calc}

\DeclareMathOperator*{\argmin}{arg\,min}

\newcommand{\Ac}{\mathcal{A}}
\newcommand{\Nc}{\mathcal{N}}

\newcommand{\x}{\mathbf{x}}

\newcommand{\y}{\mathbf{y}}

\newcommand{\z}{\mathbf{z}}

\newcommand{\Ib}{\mathbf{I}}

\newcommand{\thetab}{\mathbf{\theta}}
\newcommand{\epsilonb}{\boldsymbol{\epsilon}}
\newcommand{\mub}{\boldsymbol{\mu}}

\newcommand{\Rd}{{\mathbb R}}

\usepackage{tikz, pgfplots}
\usetikzlibrary{positioning}
\usepackage{capt-of}
\usepackage{diagbox}

\definecolor{C0}{rgb}{0.121569, 0.466667, 0.705882}
\definecolor{C1}{rgb}{1.000000, 0.498039, 0.054902}
\definecolor{C2}{rgb}{0.172549, 0.627451, 0.172549}
\definecolor{C3}{rgb}{0.839216, 0.152941, 0.156863}
\definecolor{C4}{rgb}{0.580392, 0.403922, 0.741176}
\definecolor{C5}{rgb}{0.549020, 0.337255, 0.294118}
\definecolor{C6}{rgb}{0.890196, 0.466667, 0.760784}
\definecolor{C7}{rgb}{0.498039, 0.498039, 0.498039}
\definecolor{C8}{rgb}{0.737255, 0.741176, 0.133333}
\definecolor{C9}{rgb}{0.090196, 0.745098, 0.811765}

\definecolor{trolleygrey}{rgb}{0.5, 0.5, 0.5}
\definecolor{darkpastelgreen}{rgb}{0.01, 0.75, 0.24}
\definecolor{darkpink}{rgb}{0.91, 0.33, 0.5}
\definecolor{alizarin}{rgb}{0.82, 0.1, 0.26}
\definecolor{americanrose}{rgb}{1.0, 0.01, 0.24}
\definecolor{lightblue}{rgb}{0.88, 0.88, 0.88}

%% file: appendix.tex
\appendix


\section{Score-based Diffusion Models}

Here, we review the continues form of diffusion model introduced for completeness. 
The forward process of diffusion can be formulated by an It\^{o} SDE:
\begin{equation}\label{eq:forward_sde}
    \mathrm{d}\mathbf{x} = \mathbf{f}(\mathbf{x}, t) \mathrm{d}t + g(t) \mathrm{d} \mathbf{w},
\end{equation}
where $\mathbf{f}(\cdot,t): \Rd^d \mapsto \Rd^d$ is a drift coefficient function, ${g}(t) \in \Rd$ is a scalar function known as the diffusion coefficient, and $\mathbf{w} \in \Rd^d$ is the standard Wiener process. 
The forward process of DDPM can be viewed as variance-preserving (VP) SDE~\cite{song2021scorebased} as the total variance is preserved.

Correspondingly, the reversed generative (i.e., denoising) process is given by the reverse-time SDE:
\begin{equation}\label{eq:revers_sde}
\mathrm{d}\mathbf{x} = [\mathbf{f}(\mathbf{x}, t) - g^2(t) \nabla_{\mathbf{x}_t} \log p_t(\mathbf{x}_t)]\mathrm{d}t + g(t) \mathrm{d} \mathbf{\Bar{w}},
\end{equation}
where $\mathrm{d} \mathbf{\Bar{w}}$ denotes the standard Wiener process running backward in time and $p_t(\mathbf{x}_t)$ denotes the marginal probability density w.r.t. $\mathbf{x}$ at time $t$.
In practice, a time-dependent network $\mathbf{s}_\thetab(\x_t,t)$ parameterized by $\thetab$ is trained to approximate the (Stein) score function $\nabla_{\mathbf{x}_t}\log p_t(\mathbf{x}_t)$ with score-matching method \cite{vincent2011connection}:
\begin{equation}\label{eq:objective}
\begin{aligned}
   \min_\thetab
  \mathbb{E}_{t,\mathbf{x}_0, \mathbf{x}_t}
   \big[\|\mathbf{s}_\thetab(\mathbf{x}_t, t) - \nabla_{\mathbf{x}_t}\log p(\mathbf{x}_t | \mathbf{x}_0)\|_2^2 \big],
\end{aligned}
\end{equation}
where $t$ is uniformly sampled from $[0,T]$, $\mathbf{x}_0 \sim q_{data}(\mathbf x)$ and $\mathbf{x}_t \sim q(\mathbf{x}_t | \mathbf{x}_0)$.
Once we have aceess to the well-trained $\mathbf{s}_{\thetab}(\mathbf{x}_t, t)$, an clean sample can be derived by simulating the generative reverse-time SDE~\eqref{eq:revers_sde} using numerical solvers (e.g. Euler-Maruyama).

\section{Additional Details} \label{sec:app_proof}
\subsection{Derivation for   \texorpdfstring{$\x_{t-1}^*$}.  }
We show the detailed derivation and explanation for the $\x_{t-1}^*$ in \cref{eq:findz1}. 
Our motivation was to give the stochastic sampling process a supervision at  inference, trying to make the sampled $\x_{t-1}$ in the vicinity of the theoretically derived solution.
Recall the inference distributions defined in \cite{song2020denoising}:
\begin{equation}
\begin{aligned}
    & q(\x_{t-1} | \x_t, \x_0)  \\
    = {}& {\mathcal N} (\sqrt{\bar{\alpha}_{t-1}} \x_0 + \sqrt{1-\bar{\alpha}_{t-1}-\sigma_t^2} \cdot \frac{\x_t-\sqrt{\bar{\alpha}_t} \x_0}{\sqrt{1-\bar{\alpha}_t}}, \sigma_t^2 \Ib).
\end{aligned}
\end{equation}
By setting $\sigma_t=0$, we get our deterministic denoised estimate $\x_{t-1}^*$
\begin{equation}
\begin{aligned}
    \x_{t-1}^*={} &\sqrt{\Bar{\alpha}_{t-1}} {\x}_{0}+\sqrt{1-\Bar{\alpha}_{t-1}} \cdot \frac{\x_t-\sqrt{\Bar{\alpha}_{t}}\x_0}{\sqrt{1-\Bar{\alpha}_{t}}}\\
    ={} & \frac{\sqrt{1-\Bar{\alpha}_{t-1}}}{\sqrt{1-\Bar{\alpha}_{t}}} \x_t +(\sqrt{\Bar{\alpha}_{t-1}} - \frac{\sqrt{\Bar{\alpha}_{t}} \cdot \sqrt{1-\Bar{\alpha}_{t-1}}}{\sqrt{1-\Bar{\alpha}_{t}}}) \x_0.
\end{aligned}
\end{equation}
Then  we have the approximation for $\Ac \x_{t-1}^*$:
\begin{equation}
  \Ac \x_{t-1}^* = \frac{\sqrt{1-\Bar{\alpha}_{t-1}}}{\sqrt{1-\Bar{\alpha}_{t}}} \Ac \x_t +
(\sqrt{\Bar{\alpha}_{t-1}} - \frac{\sqrt{\Bar{\alpha}_{t}} \cdot \sqrt{1-\Bar{\alpha}_{t-1}}}{\sqrt{1-\Bar{\alpha}_{t}}}) \Ac \x_0.
\end{equation}
By applying $\y = \Ac \x_0+\mathbf{n}$,
\begin{equation}
\begin{aligned}
    & \Ac \x_{t-1}^*  \\
    ={}& \frac{\sqrt{1-\Bar{\alpha}_{t-1}}}{\sqrt{1-\Bar{\alpha}_{t}}} \Ac \x_t +
(\sqrt{\Bar{\alpha}_{t-1}} - \frac{\sqrt{\Bar{\alpha}_{t}} \cdot \sqrt{1-\Bar{\alpha}_{t-1}}}{\sqrt{1-\Bar{\alpha}_{t}}}) (\y-\mathbf{n})\\
    \approx{} & \frac{\sqrt{1-\Bar{\alpha}_{t-1}}}{\sqrt{1-\Bar{\alpha}_{t}}} \Ac \x_t +
(\sqrt{\Bar{\alpha}_{t-1}} - \frac{\sqrt{\Bar{\alpha}_{t}} \cdot \sqrt{1-\Bar{\alpha}_{t-1}}}{\sqrt{1-\Bar{\alpha}_{t}}}) \y
\end{aligned}
\end{equation}

The approximation of last step holds when $\sigma_y$ is assumed in a small range as $\mathbf{n} \sim \mathcal{N}(\mathbf{0},\mathbf{\sigma_y}\Ib)$.

\begin{proposition}
    Assume that $\|\y\|_2 \leq Y, \|\x_0\|_2 \leq X, \|\x_{t-1}'\|_2 \leq X, \|\mub_\theta(\x_t, t, \y)\|_2 \leq E$ are bounded and continuous, there exists a upper bound for $\|\Ac\left(\x_{t-1}^i-\x^*_{t-1}\right)\|_2^2$ which is:
    \begin{equation}
        \|\Ac\left(\x_{t-1}^i-\x^*_{t-1}\right)\|_2^2  \leq \left(\Ac E - \Ac X + \Ac \sigma_t \right)^2.
    \end{equation}
    It is proven that some linear (i.e.~super-resolution and inpainting) operators $\Ac$ can be modeled as matrices~\cite{chung2023diffusion} which are bounded and linear. And for $\sigma_t$, it is usually set to a small value which is also bounded. Thus, we can conclude that the above upper bound holds.
\end{proposition}
\begin{proof}
    \begin{align}
            & \|\Ac\left(\x_{t-1}^i- \x^*_{t-1} \right) \|_2^2 \notag\\
        ={} & \|\Ac\x_{t-1}^i-\Ac\x^*_{t-1}\|_2^2 \\
        \overset{(\text{a})}{\approx}  & \|\Ac\left(\mub_\theta(\x_t, t, \y) + \sigma_t \z^i_t \right) -  \notag \\
        & \left(\sqrt{\bar{\alpha}_{t-1}} \y + \sqrt{1-\bar{\alpha}_{t-1}} \cdot \frac{\Ac\x_t - \sqrt{\bar{\alpha}_t}\y}{\sqrt{1 - \bar{\alpha}_t}}\right)\|_2^2,
    \end{align}
    where $\z^i_t \sim \mathcal{N}(\mathbf{0},\Ib)$ and the $\overset{(\text{a})}{\approx}$ is from Eq.~\eqref{eq:Axtappr}. Here, we apply trigonometric inequalities to the above equation, we have:
    \begin{align}
        &\|\Ac\left(\x_{t-1}^i-\x^*_{t-1}\right) \|_2^2  \notag\\
        \overset{(\text{a})}{\approx} & \|\Ac\left(\mub_\theta(\x_t, t, \y) + \sigma_t \z^i_t \right)- \notag\\& \left(\sqrt{\bar{\alpha}_{t-1}} \y + \sqrt{1-\bar{\alpha}_{t-1}} \cdot \frac{\Ac\x_t - \sqrt{\bar{\alpha}_t}\y}{\sqrt{1 - \bar{\alpha}_t}}\right)\|_2^2 \\
        \approx {}& \|\left(\Ac \mub_\theta(\x_t, t, \y) - \sqrt{\bar{\alpha}_{t-1}}\y \right) +  \notag\\& \left(\sigma_t \Ac \z_t^i - \sqrt{1-\bar{\alpha}_{t-1}} \cdot \frac{\Ac\x_t - \sqrt{\bar{\alpha}_t}\y}{\sqrt{1 - \bar{\alpha}_t}} \right)\|_2^2 \\
        \approx {}& \left(\|\Ac \mub_\theta(\x_t, t, \y) - \sqrt{\bar{\alpha}_{t-1}}\y\|_2 +  \notag \right.\\&\left. \|\sigma_t \Ac \z_t^i - \sqrt{1-\bar{\alpha}_{t-1}} \cdot \frac{\Ac\x_t - \sqrt{\bar{\alpha}_t}\y}{\sqrt{1 - \bar{\alpha}_t}}\|_2 \right)^2.
    \end{align}
    Since $\|\mub_\theta(\x_t, t, \y)\|_2 \leq E$, $\|\bar{\alpha}_t\|_2 \leq A$, $\|\bar{\alpha}_t\|_2 \leq A$, $\|\bar{\alpha}_{t-1}\|_2 \leq A,$, and $\|\y\|_2 \leq Y$, we have the upper bound for $\|\Ac \mub_\theta(\x_t, t, \y) - \sqrt{\bar{\alpha}_{t-1}}\y\|_2$, and thus $\|\Ac \mub_\theta(\x_t, t, \y) - \sqrt{\bar{\alpha}_{t-1}}\y\|_2 \leq \Ac E + \sqrt{A}Y$. Similarly, we have the upper bound for $\|\sqrt{1-\bar{\alpha}_{t-1}} \cdot \frac{\Ac\x_t - \sqrt{\bar{\alpha}_t}\y}{\sqrt{1 - \bar{\alpha}_t}}\|_2$, which is $\|\sqrt{1-\bar{\alpha}_{t-1}} \cdot \frac{\Ac\x_t - \sqrt{\bar{\alpha}_t}\y}{\sqrt{1 - \bar{\alpha}_t}}\|_2 = \|\frac{\sqrt{1-\bar{\alpha}_{t-1}} \Ac \x_t}{\sqrt{1-\bar{\alpha}_{t}}} - \frac{\sqrt{\bar{\alpha}_t(1-\bar{\alpha}_{t-1})}y}{\sqrt{1-\bar{\alpha}_{t}}}\| \leq \Ac X + \sqrt{A}Y$. Therefore, we have:
    \begin{align}
        & \|\Ac\left(\x_{t-1}^i-\x^*_{t-1}\right)\|_2^2    \notag\\
        \overset{(\text{a})}{\approx} & \|\Ac\left(\mub_\theta(\x_t, t, \y) + \sigma_t \z^i_t \right)- \notag\\& \left(\sqrt{\bar{\alpha}_{t-1}} \y + \sqrt{1-\bar{\alpha}_{t-1}} \cdot \frac{\Ac\x_t - \sqrt{\bar{\alpha}_t}\y}{\sqrt{1 - \bar{\alpha}_t}}\right)\|_2^2  \\
        \leq{} & \left(\Ac E + \Ac X + 2\sqrt{A}Y + \Ac \sigma_t \right)^2.
    \end{align}
\citet{chung2023diffusion} has shown that some linear operations $\Ac$ (such as super-resolution and inpainting) can be represented as matrices that are linear and bounded. Moreover, $\sigma_t$ is typically chosen to be a small and finite number. Therefore, we can affirm that the upper bound in the previous equation is valid.
\end{proof}

\begin{proposition}
    For the random variable $\z \sim \mathcal{N}(\mathbf{0},\Ib)$ and its objective function:
    \begin{equation}
        f(\z) = \|\Ac(\mub_{\thetab}(\x_t,t,\y)+\sigma_t \z - C_1 \x_t ) -C_2\y \|^2_2,
    \end{equation}
    where $C_1 = {\sqrt{1-\Bar{\alpha}_{t-1}}}/ {\sqrt{1-\Bar{\alpha}_{t}}}$ and $C_2=\sqrt{\Bar{\alpha}_{t-1}} - {\sqrt{\Bar{\alpha}_{t}} \sqrt{1-\Bar{\alpha}_{t-1}}}/ \\ {\sqrt{1-\Bar{\alpha}_{t}}}$.
    Thus, with $M$ trials, each consisting of $N$ samples, we have the variance for $f(\z)$, which is $\text{Var}\left(f(\z)\right)$. We have
    \begin{equation}
        \text{Var}_\text{DPS}\left(f(\z)\right) > \text{Var}_\text{MC}\left(f(\z)\right) > \text{Var}_\text{Ours}\left(f(\z)\right),
    \end{equation}
    here, $\text{Var}_\text{DPS}\left(f(\z)\right)$ is the variance of DPS~\cite{chung2023diffusion}, $\text{Var}_\text{MC}\left(f(\z)\right)$ is the variance of Monte Carlo sampling, and $\text{Var}_\text{Ours}\left(f(\z)\right)$ is the variance of our proximal sampling method. 
\end{proposition}

\begin{proof}
    For each trial, we only take one sample $\z_1$ from the distribution, and compute $f(\z_1)$. Thus, for DPS, the estimation of a single sample is $\hat{\mu}_\text{single} = f(\z_1)$, and the variance of the estimation is $\text{Var}\left(f(\z)\right)$.
    
    For Monte Carlo sampling, we draw $N$ independently and identically distributed samples $\z_1,\z_2,\dots,\z_N$, and compute the function $f(\z_i)$ for each sample. Finally,  we can get the estimation of a single trial as:
    \begin{equation}
        \hat{\mu}_\text{MC} = \frac{1}{N} \sum^{N}_{i=1} f(\z_i).
    \end{equation}
    By the central limit theorem, when $N$ is large enough, the variance of $\hat{\mu}_\text{MC}$ is $\frac{\text{Var}\left(f(\z)\right)}{N}$.
    
    For our proximal sampling method, similar to Monte Carlo sampling, we also draw $N$ independently and identically distributed samples $\z_1,\z_2,\dots,\z_N$ for each trial, and compute the function $f(\z_i)$ for each sample. However, we only get the estimation from the sample with the lowest objective function value:
    \begin{equation}
        \hat{\mu}_\text{Ours} = \argmin_\z {f(\z_1, \z_2, \dots, \z_N)}.
    \end{equation}
    In contrast to Monte Carlo sampling, which averages all function values, our proximal sampling selects only the samples corresponding to the smallest function values for each trial. Thus, the set of function value selected by our method is a subset of those sampled by Monte Carlo sampling. Obviously, for multiple trials, the variance of our method is smaller than Monte Carlo sampling, and also smaller than random sampling as in DPS. Therefore, we have:
    \begin{equation}
        \text{Var}_\text{DPS}\left(f(\z)\right) > \text{Var}_\text{MC}\left(f(\z)\right) = \frac{\text{Var}_\text{DPS}\left(f(\z)\right)}{N} > \text{Var}_\text{Ours}\left(f(\z)\right).
    \end{equation}
\end{proof}

\section{Additional Experimental Results}

\begin{figure*}[ht]
\centering
\includegraphics[width=.92\textwidth]{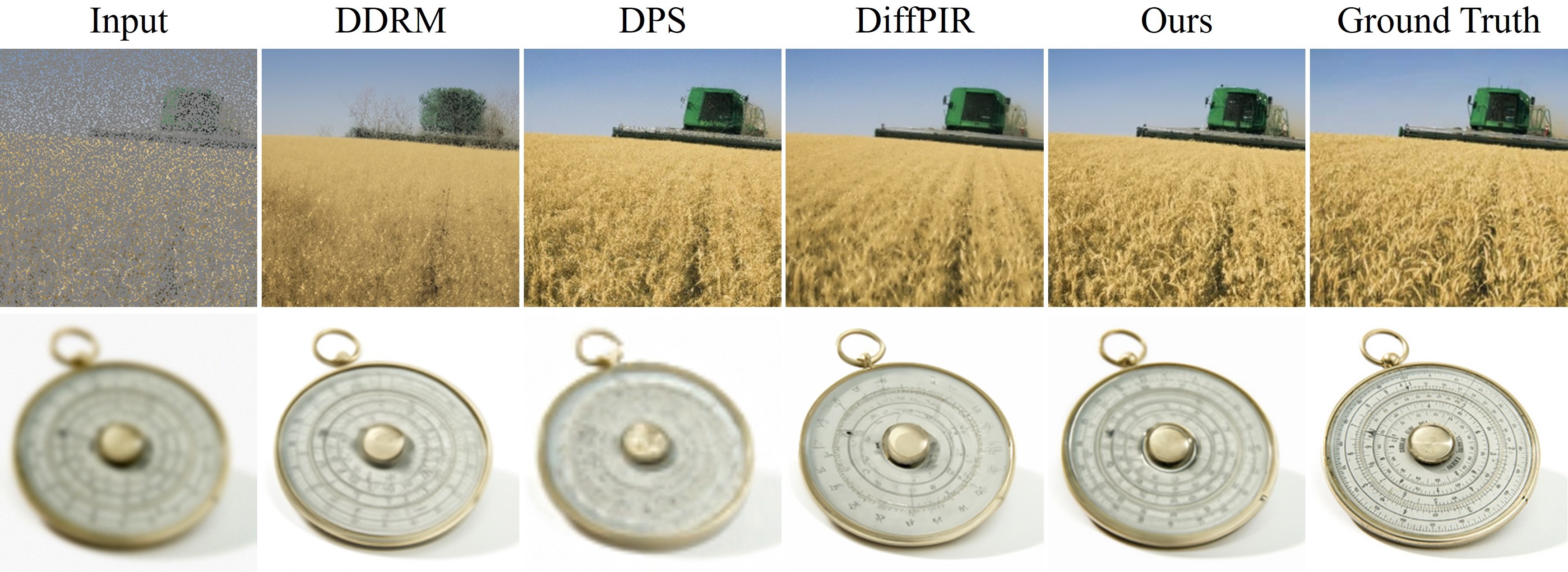}
\caption{More visual comparison of inpainting and deblurring on ImageNet.} \Description{More visual comparison of inpainting and deblurring on ImageNet.}
\label{fig:visual++}
\end{figure*}

\subsection{Less Error Accumulation}
As discussed in \cref{sec:result_selection}, sampling in proximity to the measurement yields less error accumulation. This is achieved by treating the injected noise as an adaptive correction. To verify this, \cref{fig:priordis} reports the Frobenius norm between true values $\sqrt{\bar{\alpha}_{t-1}}\x_0$ and the predicted values $\mub_\thetab(\x_t,t,\y)$.

The results validate that our proximal sampling exhibits superior predictive accuracy and less error accumulation compared to random sampling, consequently enhancing the precision in predicting subsequent samples.

\begin{figure}[htbp]
\centering
\subfigure{\includegraphics[width=0.49\textwidth]{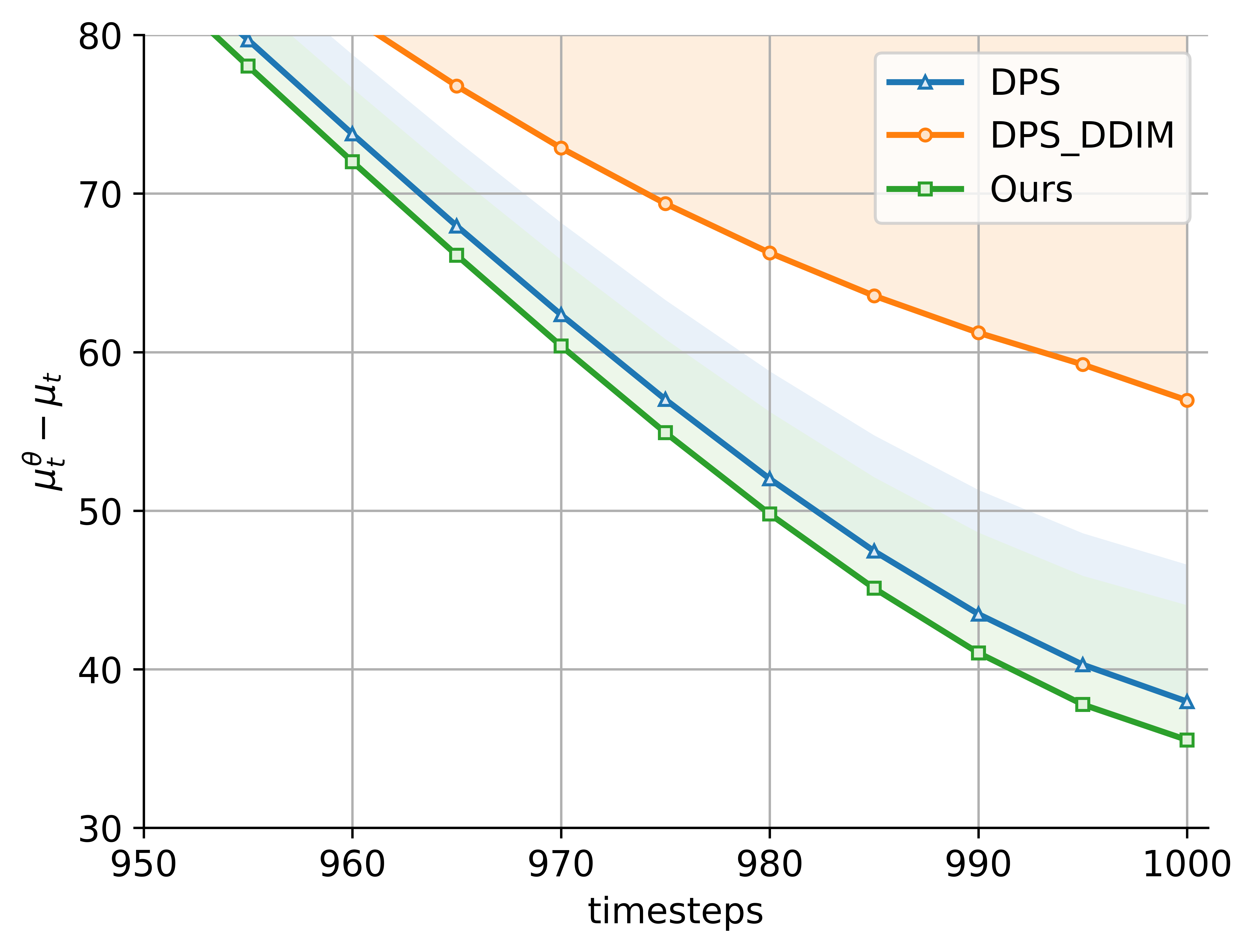}}
\caption{We report the average value of $\|\sqrt{\bar{\alpha}_{t-1}}\x_0-\mub_\thetab(\x_t,t,\y)\|_F$ on SR ($\times$4) task. The results show that our method achieves better predictive accuracy and reduces error accumulation. } \Description{We report the average value of $\|\sqrt{\bar{\alpha}_{t-1}}\x_0-\mub_\thetab(\x_t,t,\y)\|_F$ on SR ($\times$4) task. The results show that our method achieves better predictive accuracy and reduces error accumulation.}
\label{fig:priordis}
\end{figure}

\section{Experimental Details} \label{app:detail}
\subsection{Comparison Methods}
Since Score-SDE, DDRM, MCG, DPS, and DiffPIR are all pixel-based diffusion models, we used the same pre-trained checkpoint for fair comparison.

\textbf{PnP-ADMM:~} For PnP-ADMM we take the pre-trained model from DnCNN~\cite{zhang2017beyond} repository, and set $\tau=0.2$ and number of iterations to 10 for all inverse problem.

\textbf{Score-SDE:~} For Score-SDE, data consistency projection is conducted after unconditional diffusion denoising at each step. We adopt the same 
projection settings as suggested in \cite{chung2022improving}.

\textbf{MCG, DPS:~} The experimental results are derived from the source code implementation provided by \cite{chung2023diffusion} with the default parameter setting as suggested in the paper, i.e.  $  {\zeta_t} = 1/\|\y - \Ac(\hat\x_{0|t})\|$ for all inverse problem on FFHQ dataset, $ {\zeta_t} = 1/\|\y - \Ac(\hat\x_{0|t})\|$ for ImageNet SR and inpainting, $ {\zeta_t} = 0.4/\|\y - \Ac(\hat\x_{0|t})\|$ for ImageNet Gaussian deblur, and $ {\zeta_t} = 0.6/\|\y - \Ac(\hat\x_{0|t})\|$ for ImageNet motion deblur. The difference between these two methods is that MCG additionally applied 
data consistency steps as Euclidean projections onto the measurement set. 

\textbf{DDRM:~} We apply 20 NFEs DDIM~\cite{song2020denoising} sampling with $\eta = 0.85, \eta_B = 1.0$ for all experiment as suggested in the paper.

\textbf{LGD-MC:~} We follow the implementation of the algorithm in \cite{song2023loss} to use a Monte Carlo estimate of the gradient correction to amend the denoising process. The number of Monte Carlo samples is set to 20.

\textbf{DiffPIR:~} We use the original code and pre-trained models provided by \cite{zhu2023denoising}. We set the hyper-parameters consistent with the noiseless situation in the paper, i.e., SR: $\zeta = 0.3, \lambda = 6.0$ / motion deblur: $\zeta = 0.9, \lambda = 7.0$ / Gaussian deblur: $\zeta = 0.4, \lambda = 12.0$/ inpainting: $\zeta = 1.0 / \lambda = 7.0$.  Besides, we designate the sub\_1\_analytic as False in the motion deblurring task, since it directly leverages the pseudo-inverse of the fast Fourier transform, resulting in an unfair boost in performance~\cite{miljkovic2012application}.

\subsection{Parameter Setting}
Here, we list the hyper-parameter values for different tasks and
datasets in \cref{tab:Hyperparameters}. 
\begin{table}[htbp]
\centering
\caption{Hyper-parameter {$\lambda_t$} for each problem setting.}
\resizebox{.9\linewidth}{!}{%
\begin{tabular}{lcccc}
\toprule
Dataset & \multicolumn{1}{c}{{FFHQ 256x256}} & \multicolumn{1}{c}{{ImgaeNet 256x256}}  \\
\midrule
{Inpaint }  &  1.0 &  1.0    \\
{Deblur (Gaussian)}  &  1.0 &  0.5    \\
{Deblur (motion)}  &  1.0 &  0.3   \\
{SR ($\times 4$)}  &  1.0 &  1.0    \\
\bottomrule
\end{tabular}
}
\label{tab:Hyperparameters}
\end{table}

\subsection{Parameter Size and Speed}
The results for computational time and parameter sizes on the FFHQ model are presented in Table~\ref{tab:freq}. The parameter sizes of the diffusion-based methods remain consistent as the same model was used.
\begin{table}[ht] 
  \caption{Computational time and parameter size comparison}
  \label{tab:freq}
  \resizebox{\linewidth}{!}{
  \begin{tabular}{lccccccc}
    \toprule
     Method &PnP-ADMM &Score-SDE &MCG &DDRM &DPS &DiffPIR &Ours \\
    \midrule
    Parameter (M) & 0.56 & 93.56 &93.56 &93.56 &93.56 &93.56 &93.56\\
    \midrule
    Speed (s) & 1.71 &28.78 &56.41 & 4.79 &56.36 & 3.40 &57.25 \\
  \bottomrule
\end{tabular}}
\end{table}

\subsection{Source Code} Our implementation is now available at \url{https://github.com/74587887/DPPS_code}.

\section{Limitations and Future Work}

Our approach notably enhances perceptual metrics, yet it demonstrates a less substantial improvement in distortion metrics and, in certain tasks, experiences a slight decline. While this observation aligns with the perception-distortion trade-off phenomena as described in the literature~\cite{blau2018perception}, we acknowledge it as a noteworthy issue that warrants further investigation in our subsequent studies.
In addition, subsequent work of this paper aims to extend the application to diverse domains, including but not limited to medical image reconstruction.

\section{More Visual Results}
\label{sec:More_visual}
In this section, we provide supplementary visual results to show the effectiveness of our proposed method. \cref{fig:more_mc_1,fig:more_mc_2} indicate that our method produces images with better details and quality as $n$ increases.  \cref{fig:more_more_1} to \cref{fig:more_more_4} show the robustness of our method across different random seeds, in line with the claims made in the paper. \cref{fig:visual++} provides one more visual
comparison for inpainting and Gaussian deblurring tasks.

\begin{figure*}[ht]
    \centering
    \includegraphics[width=0.92\linewidth]{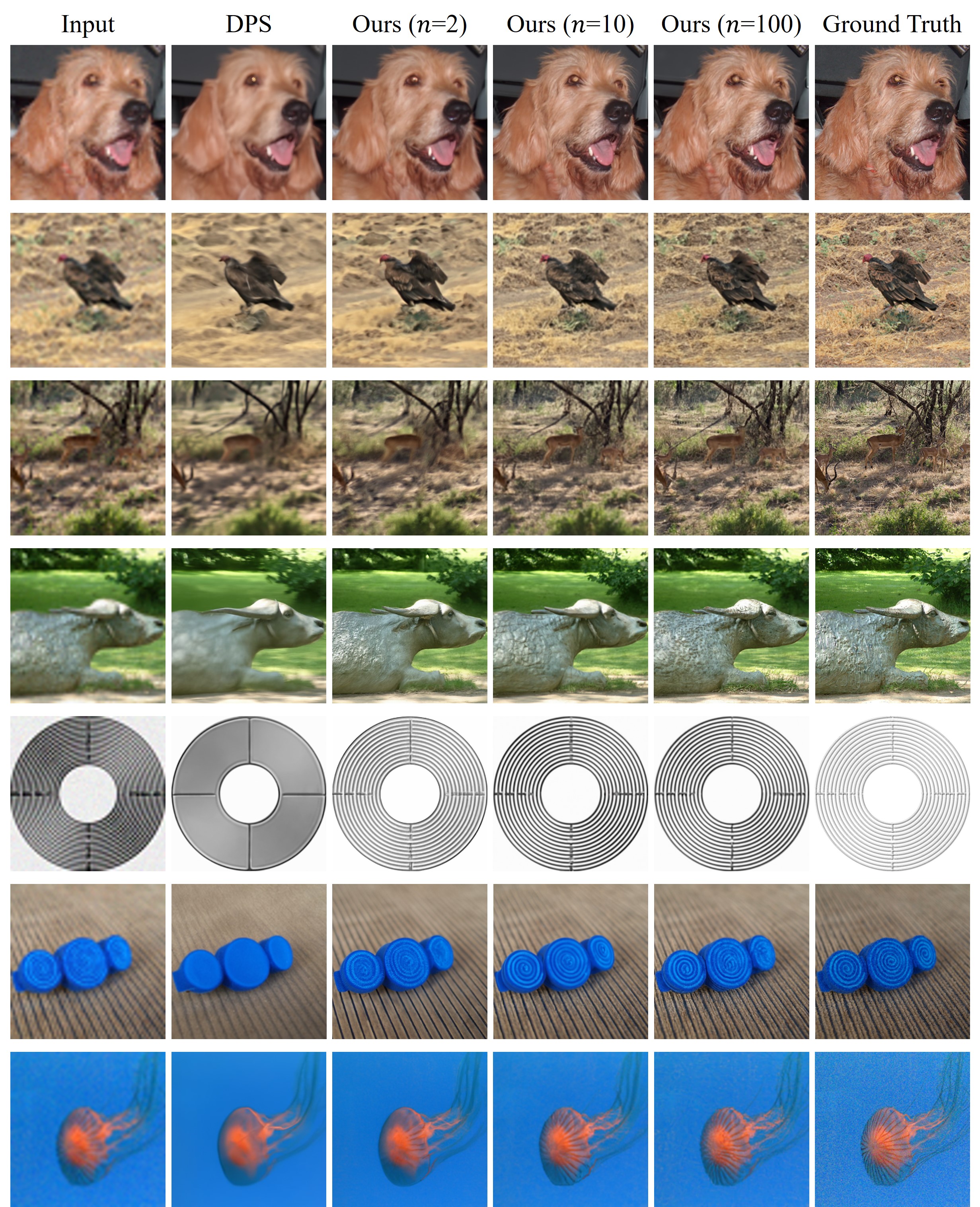}
    \caption{Qualitative results to illustrate the effectiveness of our proposed method and the impact of $n$ on SR ($\times$4) task with $\sigma_y=0.01$.}\label{fig:more_mc_1} \Description{}
\end{figure*}

\begin{figure*}[ht]
    \centering
    \includegraphics[width=0.92\linewidth]{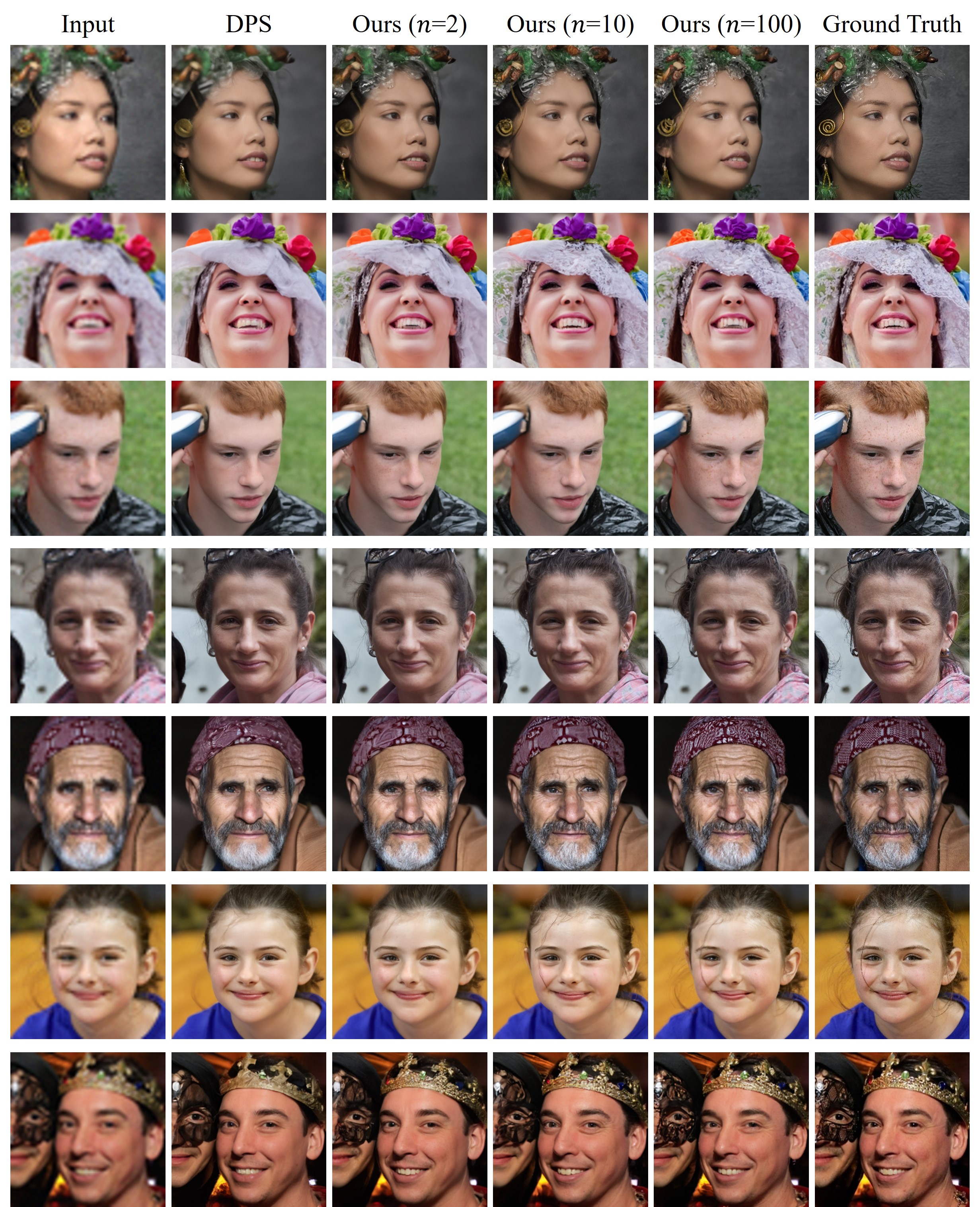}
    \caption{Qualitative results to illustrate the effectiveness of our proposed method and the impact of $n$ on SR ($\times$4) task with $\sigma_y=0.01$. }\label{fig:more_mc_2} \Description{}
\end{figure*}

\begin{figure*}[ht]
    \centering
    \includegraphics[width=0.92\linewidth]{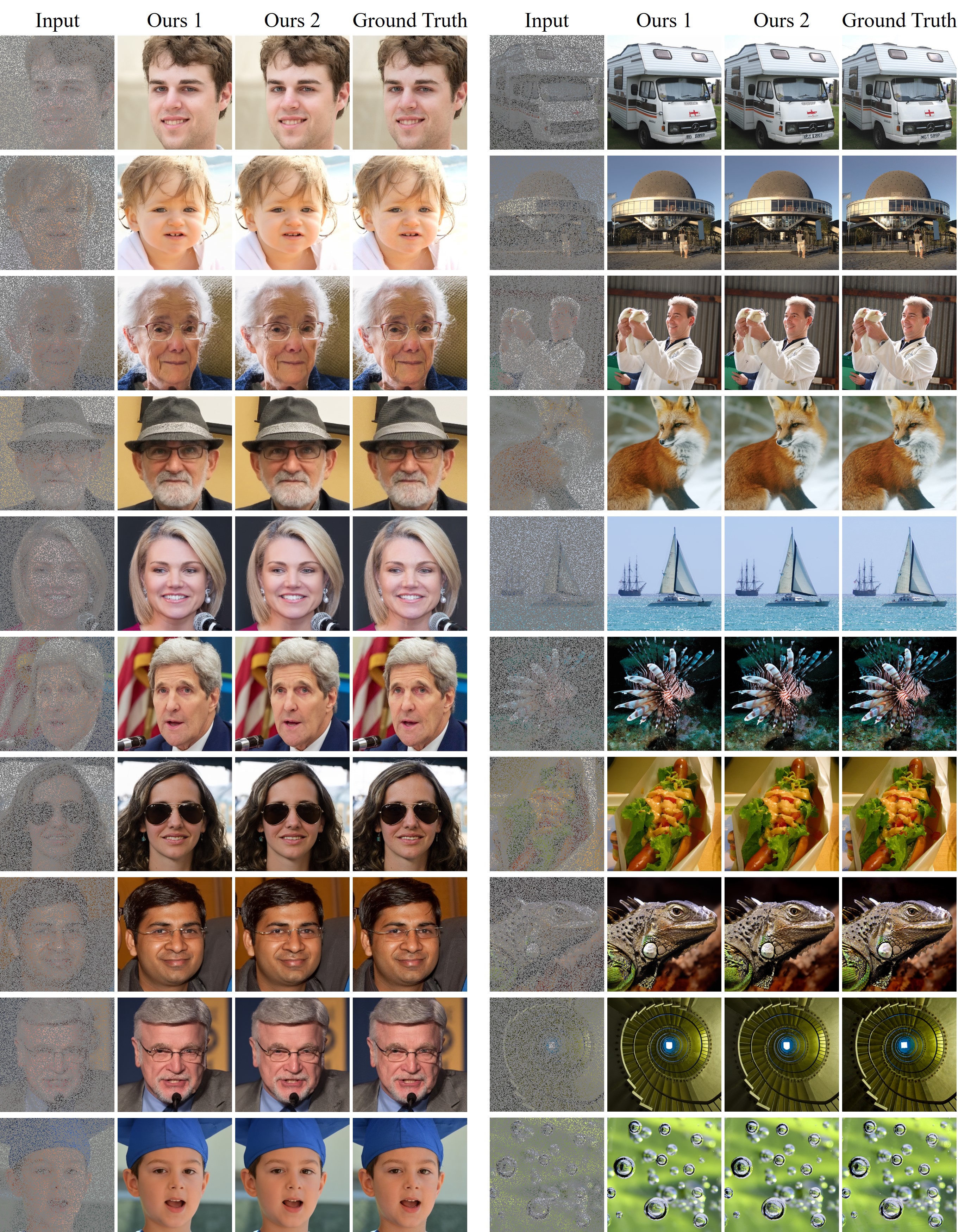}
    \caption{Qualitative inpainting results (Left FFHQ, Right ImageNet) with $\sigma_y=0.01$.}\label{fig:more_more_1} \Description{}
\end{figure*}

\begin{figure*}[ht]
    \centering
    \includegraphics[width=0.92\linewidth]{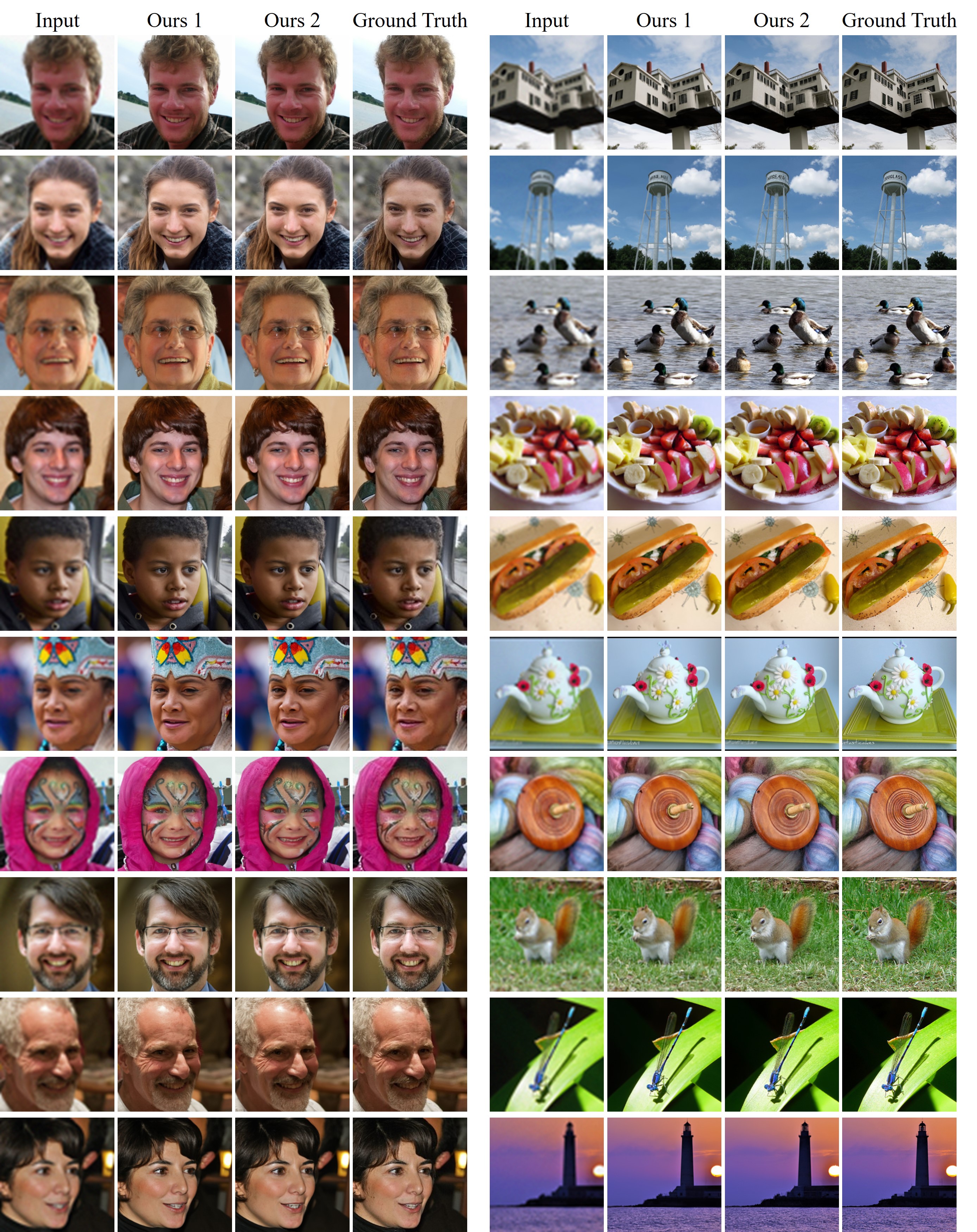}
    \caption{Qualitative SR ($\times$4) results (Left FFHQ, Right ImageNet) with $\sigma_y=0.01$.}\label{fig:more_more_2} \Description{}
\end{figure*}

\begin{figure*}[ht]
    \centering
    \includegraphics[width=0.92\linewidth]{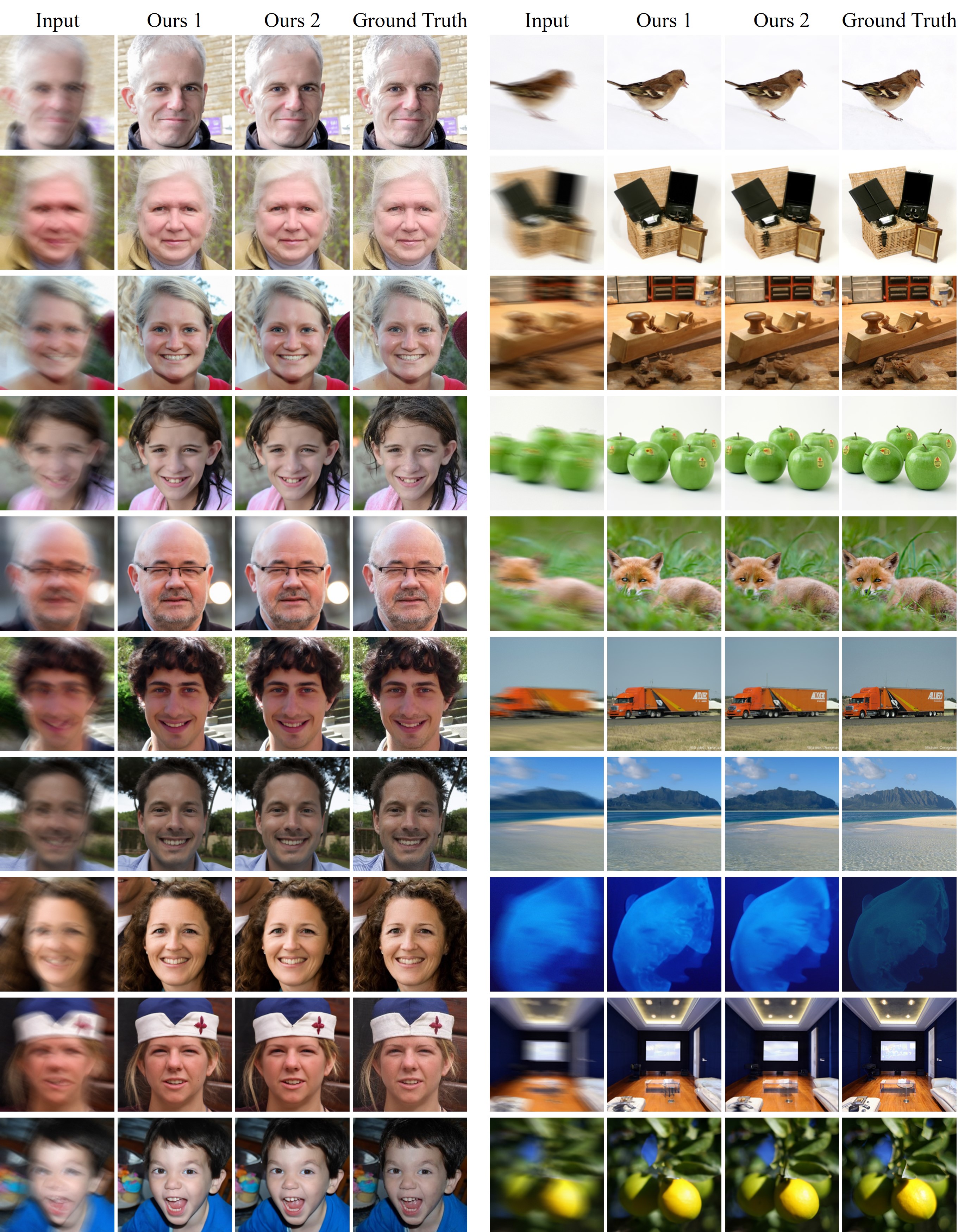}
    \caption{Qualitative motion deblurring results (Left FFHQ, Right ImageNet) with $\sigma_y=0.01$.}\label{fig:more_more_3} \Description{}
\end{figure*}

\begin{figure*}[ht]
    \centering
    \includegraphics[width=0.92\linewidth]{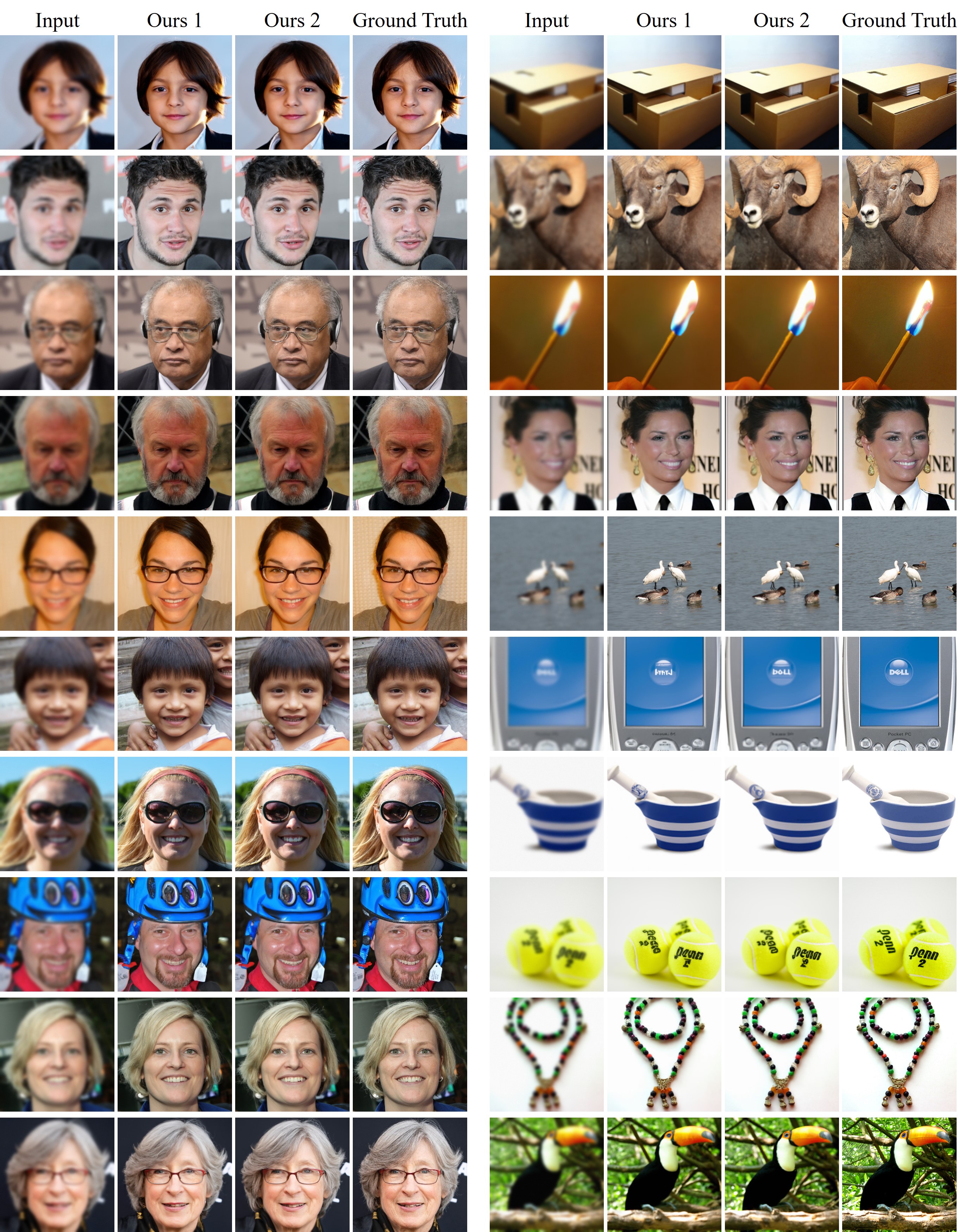}
    \caption{Qualitative Gaussian deblurring results (Left FFHQ, Right ImageNet) with $\sigma_y=0.01$.}\label{fig:more_more_4} \Description{}
\end{figure*}


%% file: sample-sigconf.bbl

\begin{thebibliography}{61}


\ifx \showCODEN    \undefined \def \showCODEN     #1{\unskip}     \fi
\ifx \showDOI      \undefined \def \showDOI       #1{#1}\fi
\ifx \showISBNx    \undefined \def \showISBNx     #1{\unskip}     \fi
\ifx \showISBNxiii \undefined \def \showISBNxiii  #1{\unskip}     \fi
\ifx \showISSN     \undefined \def \showISSN      #1{\unskip}     \fi
\ifx \showLCCN     \undefined \def \showLCCN      #1{\unskip}     \fi
\ifx \shownote     \undefined \def \shownote      #1{#1}          \fi
\ifx \showarticletitle \undefined \def \showarticletitle #1{#1}   \fi
\ifx \showURL      \undefined \def \showURL       {\relax}        \fi
\providecommand\bibfield[2]{#2}
\providecommand\bibinfo[2]{#2}
\providecommand\natexlab[1]{#1}
\providecommand\showeprint[2][]{arXiv:#2}

\bibitem[Bao et~al\mbox{.}(2021)]%
        {bao2021analytic}
\bibfield{author}{\bibinfo{person}{Fan Bao}, \bibinfo{person}{Chongxuan Li}, \bibinfo{person}{Jun Zhu}, {and} \bibinfo{person}{Bo Zhang}.} \bibinfo{year}{2021}\natexlab{}.
\newblock \showarticletitle{Analytic-DPM: an Analytic Estimate of the Optimal Reverse Variance in Diffusion Probabilistic Models}. In \bibinfo{booktitle}{\emph{International Conference on Learning Representations}}.
\newblock


\bibitem[Benton et~al\mbox{.}(2023)]%
        {benton2023linear}
\bibfield{author}{\bibinfo{person}{Joe Benton}, \bibinfo{person}{Valentin De~Bortoli}, \bibinfo{person}{Arnaud Doucet}, {and} \bibinfo{person}{George Deligiannidis}.} \bibinfo{year}{2023}\natexlab{}.
\newblock \showarticletitle{Linear convergence bounds for diffusion models via stochastic localization}.
\newblock \bibinfo{journal}{\emph{arXiv preprint arXiv:2308.03686}} (\bibinfo{year}{2023}).
\newblock


\bibitem[Blau and Michaeli(2018)]%
        {blau2018perception}
\bibfield{author}{\bibinfo{person}{Yochai Blau} {and} \bibinfo{person}{Tomer Michaeli}.} \bibinfo{year}{2018}\natexlab{}.
\newblock \showarticletitle{The perception-distortion tradeoff}. In \bibinfo{booktitle}{\emph{Proceedings of the IEEE conference on computer vision and pattern recognition}}. \bibinfo{pages}{6228--6237}.
\newblock


\bibitem[Cao et~al\mbox{.}(2023)]%
        {cao2023deep}
\bibfield{author}{\bibinfo{person}{Jiezhang Cao}, \bibinfo{person}{Yue Shi}, \bibinfo{person}{Kai Zhang}, \bibinfo{person}{Yulun Zhang}, \bibinfo{person}{Radu Timofte}, {and} \bibinfo{person}{Luc Van~Gool}.} \bibinfo{year}{2023}\natexlab{}.
\newblock \showarticletitle{Deep Equilibrium Diffusion Restoration with Parallel Sampling}.
\newblock \bibinfo{journal}{\emph{arXiv preprint arXiv:2311.11600}} (\bibinfo{year}{2023}).
\newblock


\bibitem[Chan et~al\mbox{.}(2016)]%
        {chan2016plug}
\bibfield{author}{\bibinfo{person}{Stanley~H Chan}, \bibinfo{person}{Xiran Wang}, {and} \bibinfo{person}{Omar~A Elgendy}.} \bibinfo{year}{2016}\natexlab{}.
\newblock \showarticletitle{Plug-and-play ADMM for image restoration: Fixed-point convergence and applications}.
\newblock \bibinfo{journal}{\emph{IEEE Transactions on Computational Imaging}} \bibinfo{volume}{3}, \bibinfo{number}{1} (\bibinfo{year}{2016}), \bibinfo{pages}{84--98}.
\newblock


\bibitem[Chen et~al\mbox{.}(2023)]%
        {chen2023imaging}
\bibfield{author}{\bibinfo{person}{Dongdong Chen}, \bibinfo{person}{Mike Davies}, \bibinfo{person}{Matthias~J Ehrhardt}, \bibinfo{person}{Carola-Bibiane Sch{\"o}nlieb}, \bibinfo{person}{Ferdia Sherry}, {and} \bibinfo{person}{Juli{\'a}n Tachella}.} \bibinfo{year}{2023}\natexlab{}.
\newblock \showarticletitle{Imaging With Equivariant Deep Learning: From unrolled network design to fully unsupervised learning}.
\newblock \bibinfo{journal}{\emph{IEEE Signal Processing Magazine}} \bibinfo{volume}{40}, \bibinfo{number}{1} (\bibinfo{year}{2023}), \bibinfo{pages}{134--147}.
\newblock


\bibitem[Chen et~al\mbox{.}(2021)]%
        {chen2021equivariant}
\bibfield{author}{\bibinfo{person}{Dongdong Chen}, \bibinfo{person}{Juli{\'a}n Tachella}, {and} \bibinfo{person}{Mike~E Davies}.} \bibinfo{year}{2021}\natexlab{}.
\newblock \showarticletitle{Equivariant imaging: Learning beyond the range space}. In \bibinfo{booktitle}{\emph{Proceedings of the IEEE/CVF International Conference on Computer Vision}}. \bibinfo{pages}{4379--4388}.
\newblock


\bibitem[Chung et~al\mbox{.}(2023a)]%
        {chung2023diffusion}
\bibfield{author}{\bibinfo{person}{Hyungjin Chung}, \bibinfo{person}{Jeongsol Kim}, \bibinfo{person}{Michael~Thompson Mccann}, \bibinfo{person}{Marc~Louis Klasky}, {and} \bibinfo{person}{Jong~Chul Ye}.} \bibinfo{year}{2023}\natexlab{a}.
\newblock \showarticletitle{Diffusion Posterior Sampling for General Noisy Inverse Problems}. In \bibinfo{booktitle}{\emph{International Conference on Learning Representations}}.
\newblock
\urldef\tempurl%
\url{https://openreview.net/forum?id=OnD9zGAGT0k}
\showURL{%
\tempurl}


\bibitem[Chung et~al\mbox{.}(2022)]%
        {chung2022improving}
\bibfield{author}{\bibinfo{person}{Hyungjin Chung}, \bibinfo{person}{Byeongsu Sim}, \bibinfo{person}{Dohoon Ryu}, {and} \bibinfo{person}{Jong~Chul Ye}.} \bibinfo{year}{2022}\natexlab{}.
\newblock \showarticletitle{Improving Diffusion Models for Inverse Problems using Manifold Constraints}. In \bibinfo{booktitle}{\emph{Advances in Neural Information Processing Systems}}, \bibfield{editor}{\bibinfo{person}{Alice~H. Oh}, \bibinfo{person}{Alekh Agarwal}, \bibinfo{person}{Danielle Belgrave}, {and} \bibinfo{person}{Kyunghyun Cho}} (Eds.).
\newblock
\urldef\tempurl%
\url{https://openreview.net/forum?id=nJJjv0JDJju}
\showURL{%
\tempurl}


\bibitem[Chung and Ye(2022)]%
        {chung2022score}
\bibfield{author}{\bibinfo{person}{Hyungjin Chung} {and} \bibinfo{person}{Jong~Chul Ye}.} \bibinfo{year}{2022}\natexlab{}.
\newblock \showarticletitle{Score-based diffusion models for accelerated MRI}.
\newblock \bibinfo{journal}{\emph{Medical Image Analysis}} (\bibinfo{year}{2022}), \bibinfo{pages}{102479}.
\newblock


\bibitem[Chung et~al\mbox{.}(2023b)]%
        {chung2023PrompttuningLD}
\bibfield{author}{\bibinfo{person}{Hyungjin Chung}, \bibinfo{person}{Jong~Chul Ye}, \bibinfo{person}{Peyman Milanfar}, {and} \bibinfo{person}{Mauricio Delbracio}.} \bibinfo{year}{2023}\natexlab{b}.
\newblock \showarticletitle{Prompt-tuning latent diffusion models for inverse problems}.
\newblock \bibinfo{journal}{\emph{ArXiv}}  \bibinfo{volume}{abs/2310.01110} (\bibinfo{year}{2023}).
\newblock


\bibitem[Corso et~al\mbox{.}(2023)]%
        {corso2023diffdock}
\bibfield{author}{\bibinfo{person}{Gabriele Corso}, \bibinfo{person}{Bowen Jing}, \bibinfo{person}{Regina Barzilay}, \bibinfo{person}{Tommi Jaakkola}, {et~al\mbox{.}}} \bibinfo{year}{2023}\natexlab{}.
\newblock \showarticletitle{DiffDock: Diffusion Steps, Twists, and Turns for Molecular Docking}. In \bibinfo{booktitle}{\emph{International Conference on Learning Representations (ICLR 2023)}}.
\newblock


\bibitem[Dhariwal and Nichol(2021)]%
        {dhariwal2021diffusion}
\bibfield{author}{\bibinfo{person}{Prafulla Dhariwal} {and} \bibinfo{person}{Alexander~Quinn Nichol}.} \bibinfo{year}{2021}\natexlab{}.
\newblock \showarticletitle{Diffusion Models Beat {GAN}s on Image Synthesis}. In \bibinfo{booktitle}{\emph{Advances in Neural Information Processing Systems}}, \bibfield{editor}{\bibinfo{person}{A.~Beygelzimer}, \bibinfo{person}{Y.~Dauphin}, \bibinfo{person}{P.~Liang}, {and} \bibinfo{person}{J.~Wortman Vaughan}} (Eds.).
\newblock


\bibitem[Efron(2011)]%
        {efron2011tweedie}
\bibfield{author}{\bibinfo{person}{Bradley Efron}.} \bibinfo{year}{2011}\natexlab{}.
\newblock \showarticletitle{Tweedie’s formula and selection bias}.
\newblock \bibinfo{journal}{\emph{J. Amer. Statist. Assoc.}} \bibinfo{volume}{106}, \bibinfo{number}{496} (\bibinfo{year}{2011}), \bibinfo{pages}{1602--1614}.
\newblock


\bibitem[Esser et~al\mbox{.}(2024)]%
        {esser2024scaling}
\bibfield{author}{\bibinfo{person}{Patrick Esser}, \bibinfo{person}{Sumith Kulal}, \bibinfo{person}{Andreas Blattmann}, \bibinfo{person}{Rahim Entezari}, \bibinfo{person}{Jonas M{\"u}ller}, \bibinfo{person}{Harry Saini}, \bibinfo{person}{Yam Levi}, \bibinfo{person}{Dominik Lorenz}, \bibinfo{person}{Axel Sauer}, \bibinfo{person}{Frederic Boesel}, {et~al\mbox{.}}} \bibinfo{year}{2024}\natexlab{}.
\newblock \showarticletitle{Scaling rectified flow transformers for high-resolution image synthesis}. In \bibinfo{booktitle}{\emph{Forty-first International Conference on Machine Learning}}.
\newblock


\bibitem[Everaert et~al\mbox{.}(2023)]%
        {everaert2023SignalLeak}
\bibfield{author}{\bibinfo{person}{Martin~Nicolas Everaert}, \bibinfo{person}{Athanasios Fitsios}, \bibinfo{person}{Marco Bocchio}, \bibinfo{person}{Sami Arpa}, \bibinfo{person}{Sabine S{\"u}sstrunk}, {and} \bibinfo{person}{Radhakrishna Achanta}.} \bibinfo{year}{2023}\natexlab{}.
\newblock \showarticletitle{Exploiting the Signal-Leak Bias in Diffusion Models}.
\newblock \bibinfo{journal}{\emph{arXiv preprint arXiv:2309.15842}} (\bibinfo{year}{2023}).
\newblock


\bibitem[He et~al\mbox{.}(2023)]%
        {he2023iterative}
\bibfield{author}{\bibinfo{person}{Linchao He}, \bibinfo{person}{Hongyu Yan}, \bibinfo{person}{Mengting Luo}, \bibinfo{person}{Kunming Luo}, \bibinfo{person}{Wang Wang}, \bibinfo{person}{Wenchao Du}, \bibinfo{person}{Hu Chen}, \bibinfo{person}{Hongyu Yang}, {and} \bibinfo{person}{Yi Zhang}.} \bibinfo{year}{2023}\natexlab{}.
\newblock \showarticletitle{Iterative reconstruction based on latent diffusion model for sparse data reconstruction}.
\newblock \bibinfo{journal}{\emph{arXiv preprint arXiv:2307.12070}} (\bibinfo{year}{2023}).
\newblock


\bibitem[Heusel et~al\mbox{.}(2017)]%
        {Heusel2017GANsTB}
\bibfield{author}{\bibinfo{person}{Martin Heusel}, \bibinfo{person}{Hubert Ramsauer}, \bibinfo{person}{Thomas Unterthiner}, \bibinfo{person}{Bernhard Nessler}, {and} \bibinfo{person}{Sepp Hochreiter}.} \bibinfo{year}{2017}\natexlab{}.
\newblock \showarticletitle{GANs Trained by a Two Time-Scale Update Rule Converge to a Local Nash Equilibrium}. In \bibinfo{booktitle}{\emph{Neural Information Processing Systems}}.
\newblock


\bibitem[Ho et~al\mbox{.}(2020)]%
        {ho2020denoising}
\bibfield{author}{\bibinfo{person}{Jonathan Ho}, \bibinfo{person}{Ajay Jain}, {and} \bibinfo{person}{Pieter Abbeel}.} \bibinfo{year}{2020}\natexlab{}.
\newblock \showarticletitle{Denoising diffusion probabilistic models}.
\newblock \bibinfo{journal}{\emph{Advances in Neural Information Processing Systems}}  \bibinfo{volume}{33} (\bibinfo{year}{2020}), \bibinfo{pages}{6840--6851}.
\newblock


\bibitem[Hou et~al\mbox{.}(2024)]%
        {hou2024global}
\bibfield{author}{\bibinfo{person}{Jinhui Hou}, \bibinfo{person}{Zhiyu Zhu}, \bibinfo{person}{Junhui Hou}, \bibinfo{person}{Hui Liu}, \bibinfo{person}{Huanqiang Zeng}, {and} \bibinfo{person}{Hui Yuan}.} \bibinfo{year}{2024}\natexlab{}.
\newblock \showarticletitle{Global structure-aware diffusion process for low-light image enhancement}.
\newblock \bibinfo{journal}{\emph{Advances in Neural Information Processing Systems}}  \bibinfo{volume}{36} (\bibinfo{year}{2024}).
\newblock


\bibitem[Jiang et~al\mbox{.}(2023)]%
        {jiang2023low}
\bibfield{author}{\bibinfo{person}{Hai Jiang}, \bibinfo{person}{Ao Luo}, \bibinfo{person}{Haoqiang Fan}, \bibinfo{person}{Songchen Han}, {and} \bibinfo{person}{Shuaicheng Liu}.} \bibinfo{year}{2023}\natexlab{}.
\newblock \showarticletitle{Low-light image enhancement with wavelet-based diffusion models}.
\newblock \bibinfo{journal}{\emph{ACM Transactions on Graphics (TOG)}} \bibinfo{volume}{42}, \bibinfo{number}{6} (\bibinfo{year}{2023}), \bibinfo{pages}{1--14}.
\newblock


\bibitem[Jin et~al\mbox{.}(2024)]%
        {jin2024des3}
\bibfield{author}{\bibinfo{person}{Yeying Jin}, \bibinfo{person}{Wei Ye}, \bibinfo{person}{Wenhan Yang}, \bibinfo{person}{Yuan Yuan}, {and} \bibinfo{person}{Robby~T Tan}.} \bibinfo{year}{2024}\natexlab{}.
\newblock \showarticletitle{DeS3: Adaptive Attention-Driven Self and Soft Shadow Removal Using ViT Similarity}. In \bibinfo{booktitle}{\emph{Proceedings of the AAAI Conference on Artificial Intelligence}}, Vol.~\bibinfo{volume}{38}. \bibinfo{pages}{2634--2642}.
\newblock


\bibitem[Jolicoeur-Martineau et~al\mbox{.}(2021)]%
        {jolicoeur2021gotta}
\bibfield{author}{\bibinfo{person}{Alexia Jolicoeur-Martineau}, \bibinfo{person}{Ke Li}, \bibinfo{person}{R{\'e}mi Pich{\'e}-Taillefer}, \bibinfo{person}{Tal Kachman}, {and} \bibinfo{person}{Ioannis Mitliagkas}.} \bibinfo{year}{2021}\natexlab{}.
\newblock \showarticletitle{Gotta Go Fast When Generating Data with Score-Based Models}.
\newblock \bibinfo{journal}{\emph{arXiv preprint arXiv:2105.14080}} (\bibinfo{year}{2021}).
\newblock


\bibitem[Karras et~al\mbox{.}(2022)]%
        {karras2022elucidating}
\bibfield{author}{\bibinfo{person}{Tero Karras}, \bibinfo{person}{Miika Aittala}, \bibinfo{person}{Timo Aila}, {and} \bibinfo{person}{Samuli Laine}.} \bibinfo{year}{2022}\natexlab{}.
\newblock \showarticletitle{Elucidating the Design Space of Diffusion-Based Generative Models}. In \bibinfo{booktitle}{\emph{Proc. NeurIPS}}.
\newblock


\bibitem[Karras et~al\mbox{.}(2019)]%
        {karras2019style}
\bibfield{author}{\bibinfo{person}{Tero Karras}, \bibinfo{person}{Samuli Laine}, {and} \bibinfo{person}{Timo Aila}.} \bibinfo{year}{2019}\natexlab{}.
\newblock \showarticletitle{A style-based generator architecture for generative adversarial networks}. In \bibinfo{booktitle}{\emph{Proceedings of the IEEE/CVF Conference on Computer Vision and Pattern Recognition}}. \bibinfo{pages}{4401--4410}.
\newblock


\bibitem[Kawar et~al\mbox{.}(2022)]%
        {kawar2022denoising}
\bibfield{author}{\bibinfo{person}{Bahjat Kawar}, \bibinfo{person}{Michael Elad}, \bibinfo{person}{Stefano Ermon}, {and} \bibinfo{person}{Jiaming Song}.} \bibinfo{year}{2022}\natexlab{}.
\newblock \showarticletitle{Denoising Diffusion Restoration Models}. In \bibinfo{booktitle}{\emph{Advances in Neural Information Processing Systems}}, \bibfield{editor}{\bibinfo{person}{Alice~H. Oh}, \bibinfo{person}{Alekh Agarwal}, \bibinfo{person}{Danielle Belgrave}, {and} \bibinfo{person}{Kyunghyun Cho}} (Eds.).
\newblock
\urldef\tempurl%
\url{https://openreview.net/forum?id=kxXvopt9pWK}
\showURL{%
\tempurl}


\bibitem[Kingma and Gao(2023)]%
        {kingma2023understanding}
\bibfield{author}{\bibinfo{person}{Diederik Kingma} {and} \bibinfo{person}{Ruiqi Gao}.} \bibinfo{year}{2023}\natexlab{}.
\newblock \showarticletitle{Understanding diffusion objectives as the elbo with simple data augmentation}.
\newblock \bibinfo{journal}{\emph{Advances in Neural Information Processing Systems}}  \bibinfo{volume}{36} (\bibinfo{year}{2023}).
\newblock


\bibitem[Lai et~al\mbox{.}(2023)]%
        {lai2023prox}
\bibfield{author}{\bibinfo{person}{Zeqiang Lai}, \bibinfo{person}{Kaixuan Wei}, \bibinfo{person}{Ying Fu}, \bibinfo{person}{Philipp H{\"a}rtel}, {and} \bibinfo{person}{Felix Heide}.} \bibinfo{year}{2023}\natexlab{}.
\newblock \showarticletitle{$\nabla$-prox: Differentiable proximal algorithm modeling for large-scale optimization}.
\newblock \bibinfo{journal}{\emph{ACM Transactions on Graphics (TOG)}} \bibinfo{volume}{42}, \bibinfo{number}{4} (\bibinfo{year}{2023}), \bibinfo{pages}{1--19}.
\newblock


\bibitem[Li et~al\mbox{.}(2023)]%
        {li2023alleviating}
\bibfield{author}{\bibinfo{person}{Mingxiao Li}, \bibinfo{person}{Tingyu Qu}, \bibinfo{person}{Wei Sun}, {and} \bibinfo{person}{Marie-Francine Moens}.} \bibinfo{year}{2023}\natexlab{}.
\newblock \showarticletitle{Alleviating Exposure Bias in Diffusion Models through Sampling with Shifted Time Steps}.
\newblock \bibinfo{journal}{\emph{arXiv preprint arXiv:2305.15583}} (\bibinfo{year}{2023}).
\newblock


\bibitem[Li et~al\mbox{.}(2022)]%
        {li2022diffusion}
\bibfield{author}{\bibinfo{person}{Xiang Li}, \bibinfo{person}{John Thickstun}, \bibinfo{person}{Ishaan Gulrajani}, \bibinfo{person}{Percy~S Liang}, {and} \bibinfo{person}{Tatsunori~B Hashimoto}.} \bibinfo{year}{2022}\natexlab{}.
\newblock \showarticletitle{Diffusion-lm improves controllable text generation}.
\newblock \bibinfo{journal}{\emph{Advances in Neural Information Processing Systems}}  \bibinfo{volume}{35} (\bibinfo{year}{2022}), \bibinfo{pages}{4328--4343}.
\newblock


\bibitem[Liu et~al\mbox{.}(2020)]%
        {liu2020diverse}
\bibfield{author}{\bibinfo{person}{Steven Liu}, \bibinfo{person}{Tongzhou Wang}, \bibinfo{person}{David Bau}, \bibinfo{person}{Jun-Yan Zhu}, {and} \bibinfo{person}{Antonio Torralba}.} \bibinfo{year}{2020}\natexlab{}.
\newblock \showarticletitle{Diverse image generation via self-conditioned gans}. In \bibinfo{booktitle}{\emph{Proceedings of the IEEE/CVF conference on computer vision and pattern recognition}}. \bibinfo{pages}{14286--14295}.
\newblock


\bibitem[Ma et~al\mbox{.}(2023)]%
        {ma2023OptimalBoundary}
\bibfield{author}{\bibinfo{person}{Yiyang Ma}, \bibinfo{person}{Huan Yang}, \bibinfo{person}{Wenhan Yang}, \bibinfo{person}{Jianlong Fu}, {and} \bibinfo{person}{Jiaying Liu}.} \bibinfo{year}{2023}\natexlab{}.
\newblock \showarticletitle{Solving Diffusion ODEs with Optimal Boundary Conditions for Better Image Super-Resolution}.
\newblock \bibinfo{journal}{\emph{arXiv preprint arXiv:2305.15357}} (\bibinfo{year}{2023}).
\newblock


\bibitem[Meng et~al\mbox{.}(2021)]%
        {meng2021sdedit}
\bibfield{author}{\bibinfo{person}{Chenlin Meng}, \bibinfo{person}{Yang Song}, \bibinfo{person}{Jiaming Song}, \bibinfo{person}{Jiajun Wu}, \bibinfo{person}{Jun-Yan Zhu}, {and} \bibinfo{person}{Stefano Ermon}.} \bibinfo{year}{2021}\natexlab{}.
\newblock \showarticletitle{{SDEdit}: Image Synthesis and Editing with Stochastic Differential Equations}.
\newblock \bibinfo{journal}{\emph{arXiv preprint arXiv:2108.01073}} (\bibinfo{year}{2021}).
\newblock


\bibitem[Miljkovi{\'c} et~al\mbox{.}(2012)]%
        {miljkovic2012application}
\bibfield{author}{\bibinfo{person}{Sladjana Miljkovi{\'c}}, \bibinfo{person}{Marko Miladinovi{\'c}}, \bibinfo{person}{Predrag Stanimirovi{\'c}}, {and} \bibinfo{person}{Igor Stojanovi{\'c}}.} \bibinfo{year}{2012}\natexlab{}.
\newblock \showarticletitle{Application of the pseudoinverse computation in reconstruction of blurred images}.
\newblock \bibinfo{journal}{\emph{Filomat}} \bibinfo{volume}{26}, \bibinfo{number}{3} (\bibinfo{year}{2012}), \bibinfo{pages}{453--465}.
\newblock


\bibitem[Ning et~al\mbox{.}(2023a)]%
        {ning2023elucidating}
\bibfield{author}{\bibinfo{person}{Mang Ning}, \bibinfo{person}{Mingxiao Li}, \bibinfo{person}{Jianlin Su}, \bibinfo{person}{Albert~Ali Salah}, {and} \bibinfo{person}{Itir~Onal Ertugrul}.} \bibinfo{year}{2023}\natexlab{a}.
\newblock \showarticletitle{Elucidating the Exposure Bias in Diffusion Models}.
\newblock \bibinfo{journal}{\emph{arXiv preprint arXiv:2308.15321}} (\bibinfo{year}{2023}).
\newblock


\bibitem[Ning et~al\mbox{.}(2023b)]%
        {ning2023input}
\bibfield{author}{\bibinfo{person}{Mang Ning}, \bibinfo{person}{Enver Sangineto}, \bibinfo{person}{Angelo Porrello}, \bibinfo{person}{Simone Calderara}, {and} \bibinfo{person}{Rita Cucchiara}.} \bibinfo{year}{2023}\natexlab{b}.
\newblock \showarticletitle{Input Perturbation Reduces Exposure Bias in Diffusion Models}.
\newblock \bibinfo{journal}{\emph{arXiv preprint arXiv:2301.11706}} (\bibinfo{year}{2023}).
\newblock


\bibitem[Parikh and Boyd(2014)]%
        {parikh2014proximal}
\bibfield{author}{\bibinfo{person}{Neal Parikh} {and} \bibinfo{person}{Stephen Boyd}.} \bibinfo{year}{2014}\natexlab{}.
\newblock \showarticletitle{Proximal algorithms}.
\newblock \bibinfo{journal}{\emph{Foundations and Trends in optimization}} \bibinfo{volume}{1}, \bibinfo{number}{3} (\bibinfo{year}{2014}), \bibinfo{pages}{127--239}.
\newblock


\bibitem[Rombach et~al\mbox{.}(2022)]%
        {rombach2022high}
\bibfield{author}{\bibinfo{person}{Robin Rombach}, \bibinfo{person}{Andreas Blattmann}, \bibinfo{person}{Dominik Lorenz}, \bibinfo{person}{Patrick Esser}, {and} \bibinfo{person}{Bj{\"o}rn Ommer}.} \bibinfo{year}{2022}\natexlab{}.
\newblock \showarticletitle{High-resolution image synthesis with latent diffusion models}. In \bibinfo{booktitle}{\emph{Proceedings of the IEEE/CVF Conference on Computer Vision and Pattern Recognition}}. \bibinfo{pages}{10684--10695}.
\newblock


\bibitem[Rout et~al\mbox{.}(2023a)]%
        {rout2023BeyondFT}
\bibfield{author}{\bibinfo{person}{Litu Rout}, \bibinfo{person}{Yujia Chen}, \bibinfo{person}{Abhishek Kumar}, \bibinfo{person}{Constantine Caramanis}, \bibinfo{person}{Sanjay Shakkottai}, {and} \bibinfo{person}{Wen-Sheng Chu}.} \bibinfo{year}{2023}\natexlab{a}.
\newblock \showarticletitle{Beyond First-Order Tweedie: Solving Inverse Problems using Latent Diffusion}.
\newblock \bibinfo{journal}{\emph{ArXiv}}  \bibinfo{volume}{abs/2312.00852} (\bibinfo{year}{2023}).
\newblock


\bibitem[Rout et~al\mbox{.}(2023b)]%
        {rout2023solving}
\bibfield{author}{\bibinfo{person}{Litu Rout}, \bibinfo{person}{Negin Raoof}, \bibinfo{person}{Giannis Daras}, \bibinfo{person}{Constantine Caramanis}, \bibinfo{person}{Alex Dimakis}, {and} \bibinfo{person}{Sanjay Shakkottai}.} \bibinfo{year}{2023}\natexlab{b}.
\newblock \showarticletitle{Solving Linear Inverse Problems Provably via Posterior Sampling with Latent Diffusion Models}. In \bibinfo{booktitle}{\emph{Thirty-seventh Conference on Neural Information Processing Systems}}.
\newblock


\bibitem[Russakovsky et~al\mbox{.}(2015)]%
        {russakovsky2015imagenet}
\bibfield{author}{\bibinfo{person}{Olga Russakovsky}, \bibinfo{person}{Jia Deng}, \bibinfo{person}{Hao Su}, \bibinfo{person}{Jonathan Krause}, \bibinfo{person}{Sanjeev Satheesh}, \bibinfo{person}{Sean Ma}, \bibinfo{person}{Zhiheng Huang}, \bibinfo{person}{Andrej Karpathy}, \bibinfo{person}{Aditya Khosla}, \bibinfo{person}{Michael Bernstein}, {et~al\mbox{.}}} \bibinfo{year}{2015}\natexlab{}.
\newblock \showarticletitle{Imagenet large scale visual recognition challenge}.
\newblock \bibinfo{journal}{\emph{International journal of computer vision}}  \bibinfo{volume}{115} (\bibinfo{year}{2015}), \bibinfo{pages}{211--252}.
\newblock


\bibitem[Saharia et~al\mbox{.}(2022)]%
        {saharia2022palette}
\bibfield{author}{\bibinfo{person}{Chitwan Saharia}, \bibinfo{person}{William Chan}, \bibinfo{person}{Huiwen Chang}, \bibinfo{person}{Chris Lee}, \bibinfo{person}{Jonathan Ho}, \bibinfo{person}{Tim Salimans}, \bibinfo{person}{David Fleet}, {and} \bibinfo{person}{Mohammad Norouzi}.} \bibinfo{year}{2022}\natexlab{}.
\newblock \showarticletitle{Palette: Image-to-image diffusion models}. In \bibinfo{booktitle}{\emph{ACM SIGGRAPH 2022 Conference Proceedings}}. \bibinfo{pages}{1--10}.
\newblock


\bibitem[Saharia et~al\mbox{.}(2021)]%
        {saharia2021image}
\bibfield{author}{\bibinfo{person}{Chitwan Saharia}, \bibinfo{person}{Jonathan Ho}, \bibinfo{person}{William Chan}, \bibinfo{person}{Tim Salimans}, \bibinfo{person}{David~J Fleet}, {and} \bibinfo{person}{Mohammad Norouzi}.} \bibinfo{year}{2021}\natexlab{}.
\newblock \showarticletitle{Image super-resolution via iterative refinement}.
\newblock \bibinfo{journal}{\emph{arXiv preprint arXiv:2104.07636}} (\bibinfo{year}{2021}).
\newblock


\bibitem[Singh et~al\mbox{.}(2022)]%
        {singh2022conditioning}
\bibfield{author}{\bibinfo{person}{Vedant Singh}, \bibinfo{person}{Surgan Jandial}, \bibinfo{person}{Ayush Chopra}, \bibinfo{person}{Siddharth Ramesh}, \bibinfo{person}{Balaji Krishnamurthy}, {and} \bibinfo{person}{Vineeth~N Balasubramanian}.} \bibinfo{year}{2022}\natexlab{}.
\newblock \showarticletitle{On conditioning the input noise for controlled image generation with diffusion models}.
\newblock \bibinfo{journal}{\emph{arXiv preprint arXiv:2205.03859}} (\bibinfo{year}{2022}).
\newblock


\bibitem[Sohl-Dickstein et~al\mbox{.}(2015)]%
        {sohl2015deep}
\bibfield{author}{\bibinfo{person}{Jascha Sohl-Dickstein}, \bibinfo{person}{Eric Weiss}, \bibinfo{person}{Niru Maheswaranathan}, {and} \bibinfo{person}{Surya Ganguli}.} \bibinfo{year}{2015}\natexlab{}.
\newblock \showarticletitle{Deep unsupervised learning using nonequilibrium thermodynamics}. In \bibinfo{booktitle}{\emph{International Conference on Machine Learning}}. PMLR, \bibinfo{pages}{2256--2265}.
\newblock


\bibitem[Song et~al\mbox{.}(2023a)]%
        {song2023SolvingIP}
\bibfield{author}{\bibinfo{person}{Bowen Song}, \bibinfo{person}{Soo~Min Kwon}, \bibinfo{person}{Zecheng Zhang}, \bibinfo{person}{Xinyu Hu}, \bibinfo{person}{Qing Qu}, {and} \bibinfo{person}{Liyue Shen}.} \bibinfo{year}{2023}\natexlab{a}.
\newblock \showarticletitle{Solving Inverse Problems with Latent Diffusion Models via Hard Data Consistency}.
\newblock \bibinfo{journal}{\emph{ArXiv}}  \bibinfo{volume}{abs/2307.08123} (\bibinfo{year}{2023}).
\newblock


\bibitem[Song et~al\mbox{.}(2021a)]%
        {song2020denoising}
\bibfield{author}{\bibinfo{person}{Jiaming Song}, \bibinfo{person}{Chenlin Meng}, {and} \bibinfo{person}{Stefano Ermon}.} \bibinfo{year}{2021}\natexlab{a}.
\newblock \showarticletitle{Denoising Diffusion Implicit Models}. In \bibinfo{booktitle}{\emph{9th International Conference on Learning Representations, {ICLR}}}.
\newblock


\bibitem[Song et~al\mbox{.}(2023b)]%
        {song2023pseudoinverseguided}
\bibfield{author}{\bibinfo{person}{Jiaming Song}, \bibinfo{person}{Arash Vahdat}, \bibinfo{person}{Morteza Mardani}, {and} \bibinfo{person}{Jan Kautz}.} \bibinfo{year}{2023}\natexlab{b}.
\newblock \showarticletitle{Pseudoinverse-Guided Diffusion Models for Inverse Problems}. In \bibinfo{booktitle}{\emph{International Conference on Learning Representations}}.
\newblock
\urldef\tempurl%
\url{https://openreview.net/forum?id=9_gsMA8MRKQ}
\showURL{%
\tempurl}


\bibitem[Song et~al\mbox{.}(2023c)]%
        {song2023loss}
\bibfield{author}{\bibinfo{person}{Jiaming Song}, \bibinfo{person}{Qinsheng Zhang}, \bibinfo{person}{Hongxu Yin}, \bibinfo{person}{Morteza Mardani}, \bibinfo{person}{Ming-Yu Liu}, \bibinfo{person}{Jan Kautz}, \bibinfo{person}{Yongxin Chen}, {and} \bibinfo{person}{Arash Vahdat}.} \bibinfo{year}{2023}\natexlab{c}.
\newblock \showarticletitle{Loss-Guided Diffusion Models for Plug-and-Play Controllable Generation}. In \bibinfo{booktitle}{\emph{International Conference on Machine Learning}}.
\newblock


\bibitem[Song et~al\mbox{.}(2022)]%
        {song2022solving}
\bibfield{author}{\bibinfo{person}{Yang Song}, \bibinfo{person}{Liyue Shen}, \bibinfo{person}{Lei Xing}, {and} \bibinfo{person}{Stefano Ermon}.} \bibinfo{year}{2022}\natexlab{}.
\newblock \showarticletitle{Solving Inverse Problems in Medical Imaging with Score-Based Generative Models}. In \bibinfo{booktitle}{\emph{International Conference on Learning Representations}}.
\newblock
\urldef\tempurl%
\url{https://openreview.net/forum?id=vaRCHVj0uGI}
\showURL{%
\tempurl}


\bibitem[Song et~al\mbox{.}(2021b)]%
        {song2021scorebased}
\bibfield{author}{\bibinfo{person}{Yang Song}, \bibinfo{person}{Jascha Sohl-Dickstein}, \bibinfo{person}{Diederik~P Kingma}, \bibinfo{person}{Abhishek Kumar}, \bibinfo{person}{Stefano Ermon}, {and} \bibinfo{person}{Ben Poole}.} \bibinfo{year}{2021}\natexlab{b}.
\newblock \showarticletitle{Score-Based Generative Modeling through Stochastic Differential Equations}. In \bibinfo{booktitle}{\emph{International Conference on Learning Representations}}.
\newblock
\urldef\tempurl%
\url{https://openreview.net/forum?id=PxTIG12RRHS}
\showURL{%
\tempurl}


\bibitem[Stein(1981)]%
        {stein1981estimation}
\bibfield{author}{\bibinfo{person}{Charles~M Stein}.} \bibinfo{year}{1981}\natexlab{}.
\newblock \showarticletitle{Estimation of the mean of a multivariate normal distribution}.
\newblock \bibinfo{journal}{\emph{The annals of Statistics}} (\bibinfo{year}{1981}), \bibinfo{pages}{1135--1151}.
\newblock


\bibitem[Vahdat et~al\mbox{.}(2021)]%
        {vahdat2021score}
\bibfield{author}{\bibinfo{person}{Arash Vahdat}, \bibinfo{person}{Karsten Kreis}, {and} \bibinfo{person}{Jan Kautz}.} \bibinfo{year}{2021}\natexlab{}.
\newblock \showarticletitle{Score-based generative modeling in latent space}.
\newblock \bibinfo{journal}{\emph{Advances in Neural Information Processing Systems}}  \bibinfo{volume}{34} (\bibinfo{year}{2021}), \bibinfo{pages}{11287--11302}.
\newblock


\bibitem[Vincent(2011)]%
        {vincent2011connection}
\bibfield{author}{\bibinfo{person}{Pascal Vincent}.} \bibinfo{year}{2011}\natexlab{}.
\newblock \showarticletitle{A connection between score matching and denoising autoencoders}.
\newblock \bibinfo{journal}{\emph{Neural computation}} \bibinfo{volume}{23}, \bibinfo{number}{7} (\bibinfo{year}{2011}), \bibinfo{pages}{1661--1674}.
\newblock


\bibitem[Wang et~al\mbox{.}(2023)]%
        {wang2023zeroshot}
\bibfield{author}{\bibinfo{person}{Yinhuai Wang}, \bibinfo{person}{Jiwen Yu}, {and} \bibinfo{person}{Jian Zhang}.} \bibinfo{year}{2023}\natexlab{}.
\newblock \showarticletitle{Zero-Shot Image Restoration Using Denoising Diffusion Null-Space Model}. In \bibinfo{booktitle}{\emph{The Eleventh International Conference on Learning Representations}}.
\newblock
\urldef\tempurl%
\url{https://openreview.net/forum?id=mRieQgMtNTQ}
\showURL{%
\tempurl}


\bibitem[Wei et~al\mbox{.}(2022)]%
        {wei2022tfpnp}
\bibfield{author}{\bibinfo{person}{Kaixuan Wei}, \bibinfo{person}{Angelica Aviles-Rivero}, \bibinfo{person}{Jingwei Liang}, \bibinfo{person}{Ying Fu}, \bibinfo{person}{Hua Huang}, {and} \bibinfo{person}{Carola-Bibiane Sch{\"o}nlieb}.} \bibinfo{year}{2022}\natexlab{}.
\newblock \showarticletitle{Tfpnp: Tuning-free plug-and-play proximal algorithms with applications to inverse imaging problems}.
\newblock \bibinfo{journal}{\emph{Journal of Machine Learning Research}} \bibinfo{volume}{23}, \bibinfo{number}{16} (\bibinfo{year}{2022}), \bibinfo{pages}{1--48}.
\newblock


\bibitem[Xiao et~al\mbox{.}(2021)]%
        {xiao2021tackling}
\bibfield{author}{\bibinfo{person}{Zhisheng Xiao}, \bibinfo{person}{Karsten Kreis}, {and} \bibinfo{person}{Arash Vahdat}.} \bibinfo{year}{2021}\natexlab{}.
\newblock \showarticletitle{Tackling the generative learning trilemma with denoising diffusion gans}.
\newblock \bibinfo{journal}{\emph{arXiv preprint arXiv:2112.07804}} (\bibinfo{year}{2021}).
\newblock


\bibitem[Yao et~al\mbox{.}(2024)]%
        {yao2024building}
\bibfield{author}{\bibinfo{person}{Jiawei Yao}, \bibinfo{person}{Xiaochao Pan}, \bibinfo{person}{Tong Wu}, {and} \bibinfo{person}{Xiaofeng Zhang}.} \bibinfo{year}{2024}\natexlab{}.
\newblock \showarticletitle{Building lane-level maps from aerial images}. In \bibinfo{booktitle}{\emph{ICASSP 2024-2024 IEEE International Conference on Acoustics, Speech and Signal Processing (ICASSP)}}. IEEE, \bibinfo{pages}{3890--3894}.
\newblock


\bibitem[Zhang et~al\mbox{.}(2017)]%
        {zhang2017beyond}
\bibfield{author}{\bibinfo{person}{Kai Zhang}, \bibinfo{person}{Wangmeng Zuo}, \bibinfo{person}{Yunjin Chen}, \bibinfo{person}{Deyu Meng}, {and} \bibinfo{person}{Lei Zhang}.} \bibinfo{year}{2017}\natexlab{}.
\newblock \showarticletitle{{Beyond a gaussian denoiser: Residual learning of deep CNN for image denoising}}.
\newblock \bibinfo{journal}{\emph{IEEE transactions on image processing}} \bibinfo{volume}{26}, \bibinfo{number}{7} (\bibinfo{year}{2017}), \bibinfo{pages}{3142--3155}.
\newblock


\bibitem[Zhang et~al\mbox{.}(2018)]%
        {Zhang2018TheUE}
\bibfield{author}{\bibinfo{person}{Richard Zhang}, \bibinfo{person}{Phillip Isola}, \bibinfo{person}{Alexei~A. Efros}, \bibinfo{person}{Eli Shechtman}, {and} \bibinfo{person}{Oliver Wang}.} \bibinfo{year}{2018}\natexlab{}.
\newblock \showarticletitle{The Unreasonable Effectiveness of Deep Features as a Perceptual Metric}.
\newblock \bibinfo{journal}{\emph{2018 IEEE/CVF Conference on Computer Vision and Pattern Recognition}} (\bibinfo{year}{2018}), \bibinfo{pages}{586--595}.
\newblock


\bibitem[Zhu et~al\mbox{.}(2023)]%
        {zhu2023denoising}
\bibfield{author}{\bibinfo{person}{Yuanzhi Zhu}, \bibinfo{person}{Kai Zhang}, \bibinfo{person}{Jingyun Liang}, \bibinfo{person}{Jiezhang Cao}, \bibinfo{person}{Bihan Wen}, \bibinfo{person}{Radu Timofte}, {and} \bibinfo{person}{Luc Van~Gool}.} \bibinfo{year}{2023}\natexlab{}.
\newblock \showarticletitle{Denoising Diffusion Models for Plug-and-Play Image Restoration}. In \bibinfo{booktitle}{\emph{Proceedings of the IEEE/CVF Conference on Computer Vision and Pattern Recognition}}. \bibinfo{pages}{1219--1229}.
\newblock


\end{thebibliography}
